\documentclass[11pt,english]{article}
\usepackage[T1]{fontenc}
\usepackage[latin9]{luainputenc}
\usepackage{color}
\usepackage{babel}
\usepackage{float}
\usepackage{amsmath}
\usepackage{amsthm}
\usepackage{amssymb}
\usepackage{geometry}
\usepackage{wasysym}
\usepackage[pdfusetitle,
 bookmarks=true,bookmarksnumbered=false,bookmarksopen=false,
 breaklinks=false,pdfborder={0 0 0},pdfborderstyle={},backref=page,colorlinks=true]
 {hyperref}
\hypersetup{
 linkcolor=magenta, urlcolor=marineblue, citecolor=blue}

\makeatletter

\floatstyle{ruled}
\newfloat{algorithm}{tbp}{loa}
\providecommand{\algorithmname}{Algorithm}
\floatname{algorithm}{\protect\algorithmname}

\theoremstyle{plain}
\newtheorem{thm}{\protect\theoremname}
\theoremstyle{plain}
\newtheorem{cor}[thm]{\protect\corollaryname}
\theoremstyle{plain}
\newtheorem{fact}[thm]{\protect\factname}
\theoremstyle{plain}
\newtheorem{lem}[thm]{\protect\lemmaname}
\theoremstyle{definition}
\newtheorem{defn}[thm]{\protect\definitionname}

\@ifundefined{date}{}{\date{}}
\usepackage{algorithm,algorithmicx,algpseudocode}

\makeatother

\providecommand{\corollaryname}{Corollary}
\providecommand{\definitionname}{Definition}
\providecommand{\factname}{Fact}
\providecommand{\lemmaname}{Lemma}
\providecommand{\theoremname}{Theorem}

\begin{document}
\title{Breaking the Barrier of Self-Concordant Barriers:\\
Faster Interior Point Methods for M-Matrices}
\author{Adrian Vladu\thanks{CNRS \& IRIF, Universit\'{e} Paris Cit\'{e}, \texttt{vladu@irif.fr}}}
\maketitle
\begin{abstract}
We study two fundamental optimization problems: (1) scaling a symmetric
positive definite matrix by a positive diagonal matrix so that the
resulting matrix has row and column sums equal to 1; and (2) minimizing
a quadratic function subject to hard non-negativity constraints. Both
problems lend themselves to efficient algorithms based on interior
point methods (IPMs). For general instances, standard self-concordance
theory places a limit on the iteration complexity of these methods
at $\widetilde{O}\left(n^{1/2}\right)$, where $n$ denotes the matrix
dimension. We show via an amortized analysis that, when the input
matrix is an M-matrix, an IPM with adaptive step sizes solves both
problems in only $\widetilde{O}\left(n^{1/3}\right)$ iterations.
As a corollary, using fast Laplacian solvers, we obtain an $\ell_{2}$
flow diffusion algorithm with depth $\widetilde{O}\left(n^{1/3}\right)$
and work $\widetilde{O}$$\left(n^{1/3}\cdot\text{nnz}\right)$. This
result marks a significant instance in which a standard log-barrier
IPM permits provably fewer than $\Theta\left(n^{1/2}\right)$ iterations.

\newpage{}
\end{abstract}
\global\long\def\cong#1#2{\rho_{#1,#2}}\global\long\def\update{\widetilde{\rho}}\global\long\def\diag#1{\textnormal{diag}\left(#1\right)}\global\long\def\vzeta{\zeta}\global\long\def\scalability{\boldsymbol{\mathit{\nu}}}\global\long\def\nocong#1#2{\widehat{\rho}_{#1,#2}}\global\long\def\congsimple{\rho}

 \renewcommand{\epsilon}{\varepsilon}

\section{Introduction}

Interior point methods (IPMs) are a cornerstone of modern optimization,
prized for their versatility and effectiveness across a range of applications.
Over recent years, IPMs have driven a surge of theoretical breakthroughs,
particularly for classical combinatorial optimization problems. The
core appeal of IPMs lies in their unmatched convergence guarantee,
broad applicability, and relative ease of implementation and analysis,
as they primarily rely on solving linear systems. This is especially
advantageous when these systems are structured, allowing the use of
fast linear system solvers \cite{madry2013navigating,madry2016computing,cmsv17,liu2019faster,liu2020faster,van2020bipartite,van2020solving,axiotis2020circulation,gao2021fully,van2021minimum,axiotis2022faster,van2022faster,chen2022maximum}.
Despite their straightforward structure, recent advances based on
IPMs have often required complex technical modifications to achieve
improvements, undercutting the simplicity that initially makes IPMs
attractive.

A particularly frustrating gap exists between the theoretical iteration
complexity of IPMs and their practical performance. While classical
self-concordance theory, pioneered by Nesterov and Nemirovski \cite{nesterov1994interior},
establishes an iteration bound of approximately $\sqrt{n}$ for IPMs,
where $n$ typically represents the number of constraints\footnote{In the case of linear programming, by using certain barrier functions,
the number of iterations can instead be bounded by square root of
the number of variables \cite{nesterov1994interior,lee2014path,hildebrand2014canonical,bubeck2014entropic},
but this is still $\widetilde{O}\left(\sqrt{n}\right)$ in the worst
case.}, empirical results consistently reveal that far fewer iterations
are needed in practice \cite{gondzio2012interior}. This discrepancy
raises the intriguing possibility of achieving faster algorithms through
refined analysis, without complex algorithmic modifications. This
is especially compelling for structured optimization problems, such
as flow problems on graphs, where efficient linear system solvers,
like Laplacian solvers, are available. Intriguingly, Todd and Ye \cite{todd2020lower,todd1996lower}
provide a lower bound by showing that there are linear programming
instances where $\Omega\left(n^{1/3}\right)$ iterations are required
to reduce the duality gap by a constant factor. This raises the question
whether the upper bound can be tightened, which would demand new ideas
beyond classical theory.

However, efforts to bridge this gap have been limited. In a few instances,
such as geometric median \cite{cohen2016geometric} and $\ell_{p}$
regression \cite{bubeck2018homotopy}, progress has been made by diverging
substantially from the traditional log-barrier approach. Other recent
improvements rely on aggressive perturbations of the underlying problem,
exploiting highly specific structural properties \cite{madry2013navigating,madry2016computing,cmsv17,liu2019faster,axiotis2020circulation,liu2020faster},
or only achieve lower complexity per iteration by incorporating sophisticated
dynamic data structures to speed up individual iterations \cite{cohen2021solving,chen2022maximum}.

In this work, we address this challenge by seeking fast algorithms
through a straightforward IPM approach. Specifically, we focus on
optimization problems involving \textit{M-matrices}, which are fundamental
in numerical analysis due to their desirable algebraic properties
\cite{johnson1982inverse,windisch2013m}. M-matrices are essentially
rescalings of diagonally dominant matrices, and are characterized
by non-positive off-diagonal entries and a structure that guarantees
positive definiteness under certain conditions. They arise in a variety
of applications, including Markov chains, finite difference methods,
and network systems, where they naturally generalize graph Laplacians.

Our main results concern two fundamental optimization problems involving
M-matrices:

\paragraph{Matrix Scaling (\texttt{MS}).}
\begin{itemize}
\item \textbf{Input:} A symmetric M-matrix $A\in\mathbb{R}^{n\times n}$
and a scalar parameter $\varepsilon>0$.
\item \textbf{Output:} A positive vector $x\in\mathbb{R}^{n}$ such that
$\left\Vert XAX1-1\right\Vert _{2}\leq\varepsilon$, where $X=\diag x$
and $1$ is the all-ones vector.
\end{itemize}

\paragraph{Quadratic Optimization (\texttt{QO}).}
\begin{itemize}
\item \textbf{Input:} A symmetric M-matrix $A\in\mathbb{R}^{n\times n}$,
a vector $b\in\mathbb{R}^{n}$, and a scalar parameter $\varepsilon>0$.
\item \textbf{Output:} A positive vector $x$ such that $\frac{1}{2}x^{\top}Ax-b^{\top}x\leq\min_{x^{*}\geq0}\frac{1}{2}\left(x^{*}\right)^{\top}Ax^{*}-b^{\top}x^{*}+\varepsilon$.
\end{itemize}
\texttt{MS} has been studied by Khachiyan and Kalantari in the context
of solving linear programs \cite{khachiyan1992diagonal}, as one can
easily show that, for general positive semidefinite matrices, this
instance of matrix scaling suffices to decide whether a polytope described
by linear inequalities is non-empty, which is well known to be equivalent
to linear programming \cite{jin2006procedure}.\footnote{While this reduction works well with high-precision solvers, since
it does not significantly amplify errors, the resulting problem requires
scaling a general positive semidefinite matrix, where our solver does
not offer a provable improvement. More precisely, based on the feasibility
problem $Cx\leq d$, the reduction constructs the positive semidefinite
matrix $\left[\begin{array}{cc}
CC^{\top}+dd^{\top} & d\\
d^{\top} & 1
\end{array}\right]$ which is not an M-matrix. As in the original paper, this reduction
only requires a bound on the scaling error in the $\ell_{\infty}$
norm. However, we impose a tighter bound in the $\ell_{2}$ norm instead,
since the output is used as initialization for \texttt{QO, }and it
provides cleaner guarantees in this form.} It is fundamentally different from another matrix scaling problem
encountered in literature, which given as input a non-negative matrix
$A$, seeks two different positive vectors $x$ and $y$ which make
$\diag xA\diag y$ doubly stochastic \cite{cohen2017matrix,allen2017much}.
This latter problem is significantly easier, as it can be solved by
reparametrizing $x=\exp\left(\widetilde{x}\right)$ and $y=\exp\left(\widetilde{y}\right)$,
and minimizing a convex objective in $\left(\widetilde{x};\widetilde{y}\right)$
whose first order optimality conditions precisely correspond to the
desired scaling conditions. The fast algorithms from the referenced
works do not carry over to our setting.

\texttt{QO} dualizes into a meaningful combinatorial form. Suppose
that $A$ is a Laplacian matrix that factorizes as $A=B^{\top}R^{-1}B$,
where $B$ is the incidence matrix of a graph $G\left(V,E\right)$,
and $R$ is a positive diagonal representing edge resistances.\footnote{This is not exactly an M-matrix, since it has a null space, so it
does not directly fit into our framework. This problem can be easily
fixed by perturbing the input with a tiny multiple of the identity.} Using standard duality one can verify that 
\[
\max_{x\geq0}b^{\top}x-\frac{1}{2}x^{\top}Ax=\min_{f:B^{\top}f\geq d}\frac{1}{2}\sum_{e\in E}r_{e}f_{e}^{2}\,.
\]
As shown by Fountoulakis-Wang-Yang \cite{fountoulakis2020p}, this
formulation captures an electrical flow problem where only source
demands are specified, and flow can drain freely as long as it does
not not exceed the capacity of the sinks it chooses. A physical interpretation
is that paint spills from the source nodes, flowing across the graph,
with each node serving as a sink that can absorb up to a specified
amount. 

This problem admits a highly efficient algorithm, achieving nearly
linear runtime as demonstrated by Chen-Peng-Wang \cite{chen20222}.
However, it heavily leverages the combinatorial structure of the graph,
and it is unclear to what extent it parallelizes. In fact, while the
\textit{easier} problem of solving Laplacian systems can be efficiently
parallelized, it is done using entirely different techniques \cite{peng2014efficient,kyng2016sparsified,sachdeva2023simple}
unrelated to the j-tree sparsifiers of \cite{chen20222}, which do
not directly carry over to the diffusion problem. 

In this paper, through a novel amortized analysis, we demonstrate
that for these M-matrix-based problems, a standard interior point
method with adaptively chosen step sizes can achieve a significantly
reduced iteration complexity of $\widetilde{O}\left(n^{1/3}\right)$.
This marks a notable step forward, showing that for a broad and practically
relevant class of problems, simple IPMs can indeed yield rapid convergence
with minimal technical adjustments.

\subsection{Our Results}

Our main results are the following theorems, which characterize the
convergence rate of a standard log-barrier interior point method.
We present our interior point method in the predictor-corrector framework,
where each iteration consists of a predictor step that seeks to advance
along the central path, followed by a short sequence of corrector
steps that restore centrality. While in principle, only a constant
number of corrector steps suffice to maintain sufficient centrality,
we simplify the analysis by running them until centrality is restored
to machine precision. This introduces an additional $O\left(\log\log\varepsilon_{\text{mach}}^{-1}\right)$
factor in the number of linear system solves, which we omit in our
asymptotic bounds for readability. However, all our results remain
valid even if corrector steps are limited to a constant number per
iteration.
\begin{thm}
\label{thm:main-scaling}Given an instance of the matrix scaling problem
\texttt{MS$(A,\varepsilon)$} consisting of symmetric M-matrix $A\in\mathbb{R}^{n\times n}$
and a scalar $\varepsilon>0$, an interior point method can compute
a non-negative vector $x$ such that 
\[
\left\Vert XAX1-1\right\Vert _{2}\leq\varepsilon
\]
 in $O\left(n^{1/3}\log\frac{\left\Vert A1-1\right\Vert _{2}}{\varepsilon}\right)$
predictor-corrector iterations.
\end{thm}

\begin{thm}
\label{thm:main-quadratic}Given an instance of the quadratic optimization
problem \texttt{QO$(A,b,\varepsilon)$} consisting of a symmetric
M-matrix $A\in\mathbb{R}^{n\times n}$ , a vector $b\in\mathbb{R}^{n\times n}$
and a scalar $\varepsilon>0$, an interior point method can compute
a non-negative vector $x$ such that 
\[
\frac{1}{2}x^{\top}Ax-b^{\top}x\leq\min_{x^{*}\geq0}\frac{1}{2}\left(x^{*}\right)^{\top}Ax^{*}-b^{\top}x^{*}+\varepsilon
\]
in $O\left(n^{1/3}\log\frac{\left\Vert A1-1\right\Vert _{2}+\left\Vert b\right\Vert _{2}}{\varepsilon}\right)$
predictor-corrector iterations.
\end{thm}

The underlying algorithms are extremely simple, and are described
in the following pseudocode.

\begin{algorithm}[H]
\begin{algorithmic}[1]

\Require Symmetric M-matrix $A\in\mathbb{R}^{n\times n}$, target
centrality parameter $\mu_{F}$.

\Ensure Returns a vector $x>0$ such that $\left\Vert XAX1-1\right\Vert _{2}\leq\frac{\left\Vert A1-1\right\Vert _{2}}{\sqrt{\mu_{F}}}$.

\Procedure{MS-IPM}{$A,\mu_{F}$}

\State$b\gets A1-1,\mu\gets1,x\gets1$

\While{$\mu<\mu_{F}$}

\State$\congsimple\gets\left(\frac{1}{\mu}XAX+I\right)^{-1}1$, $\delta\gets\frac{1}{32\left\Vert \congsimple\right\Vert _{3}}$,
$x\gets x\left(1+\delta\congsimple\right)$, $\mu\gets\frac{\mu}{1-\delta}$\Comment{Compute a predictor step.}

\For{$i=1,\dots,\log\log\varepsilon_{\text{mach}}^{-1}$} \Comment{Correct to centrality, up to machine precision.}

\State $x\gets x\left(1+\left(\frac{1}{\mu}XAX+I\right)^{-1}\left(1-\frac{1}{\mu}X\left(Ax-b\right)\right)\right)$

\EndFor

\EndWhile

\State\Return$x/\sqrt{\mu}$

\EndProcedure

\end{algorithmic}

\medskip{}

\caption{Scaling IPM.\label{alg:scaling-ipm}}
\end{algorithm}

\begin{algorithm}[H]
\begin{algorithmic}[1]

\Require Symmetric M-matrix $A\in\mathbb{R}^{n\times n}$, vector
$b\in\mathbb{R}^{n}$, target error $\varepsilon>0$.

\Ensure Returns a vector $x>0$ such that $\frac{1}{2}x^{\top}Ax-b^{\top}x\leq\min_{x^{*}\geq0}\frac{1}{2}\left(x^{*}\right)^{\top}Ax^{*}-b^{\top}x^{*}+\varepsilon$.

\Procedure{QO-IPM}{$A,b,\varepsilon$}

\State$\mu\gets2\left\Vert A1-1-b\right\Vert _{2}$, $x\gets$MS-IPM$(A,\mu)$.

\For{$i=1,\dots,\log\log\varepsilon_{\text{mach}}^{-1}$} \Comment{Correct to centrality, up to machine precision.}

\State $x\gets x\left(1-\left(\frac{1}{\mu}XAX+I\right)^{-1}\left(\frac{1}{\mu}X\left(Ax-b\right)-1\right)\right)$

\EndFor

\While{$\mu>\varepsilon/n$}

\State Via line-search find $\delta>0$ so that $\congsimple_{\delta}\gets\left(\frac{1+\delta}{\mu}XAX+I\right)^{-1}1$
satisfies $\frac{1}{32}\leq\delta\left\Vert \congsimple_{\delta}\right\Vert _{3}\leq\frac{1}{16}$.

\State $x\gets x\left(1-\delta\congsimple_{\delta}\right)$, $\mu\gets\frac{\mu}{1+\delta}$\Comment{Compute a predictor step.}

\For{$i=1,\dots,\log\log\varepsilon_{\text{mach}}^{-1}$} \Comment{Correct to centrality, up to machine precision.}

\State $x\gets x\left(1-\left(\frac{1}{\mu}XAX+I\right)^{-1}\left(\frac{1}{\mu}X\left(Ax-b\right)-1\right)\right)$

\EndFor

\EndWhile

\State\Return$x$

\EndProcedure

\end{algorithmic}

\medskip{}

\caption{Quadratic IPM.\label{alg:quadratic-ipm}}
\end{algorithm}

Interestingly, the scaling IPM increases the centrality parameter
$\mu$, whereas the quadratic IPM decreases it. We crucially use the
scaling IPM to construct an initial centered solution for the quadratic
IPM.

Regarding the operations executed in each step of the IPM, we note
that they consist of scaling the current iterate by solutions of linear
systems involving matrices of the type $\frac{1}{\mu}XAX+I$. As we
will further elaborate in Section \ref{subsec:m-matrix-prop} these
are also M-matrices. Thus they admit fast linear system solvers which
run in time proportional to the sparsity of $A$, with only a logarithmic
dependence on their condition number \cite{daitch2008faster,ahmadinejad2019perron}.
Since the M-matrix solver of \cite{ahmadinejad2019perron} readily
provides a rescaling that makes the input matrix diagonally dominant
we can just call the solver once, and reuse the returned rescaling
across all iterations of the IPM, where we can use a Laplacian solver
for all the system solves. This enables us to conclude that each iteration
of the IPM can be executed in nearly linear time in input sparsity.\footnote{For the sake of the presentation we assume that the M-matrix solver
provides an exact solution. In reality, it produces an approximate
solution, but it is folklore that these errors are well tolerated
by IPMs.} We provide the full analysis in Section \ref{sec:mmatrix-solver}.
\begin{cor}
\label{cor:runtime}Given an instance \texttt{MS$(A,\varepsilon)$}
one can compute a solution in time 
\begin{align*}
\widetilde{O}\left(n^{1/3}\cdot\text{nnz}\left(A\right)\cdot\log\frac{1}{\varepsilon}\right)\,,
\end{align*}
where $\widetilde{O}$ suppresses factors of $\log\left(\max\left\{ n,\kappa\left(A\right),\left\Vert A1\right\Vert _{2}\right\} \right)$
. Similarly, given an instance \texttt{QO$(A,b,\varepsilon)$ }one
can compute a solution in time
\[
\widetilde{O}\left(n^{1/3}\cdot\text{nnz}\left(A\right)\cdot\log\frac{1}{\varepsilon}\right)\,,
\]
where $\widetilde{O}$ suppresses polylogarithmic factors of $\log\left(\max\left\{ n,\kappa\left(A\right),\left\Vert A1\right\Vert _{2},\left\Vert b\right\Vert _{2}/\lambda_{\min}\left(A\right)\right\} \right)$.
\end{cor}

\subsection{Notation, Terminology, and Overview of Our Interior Point Methods}

The algorithm we analyze is the predictor-corrector IPM. It is completely
generic, and applies to any positive semidefinite matrix given as
input. Only the amortized analyses which we give in sections \ref{subsec:Improved-analysis-potential-fn}
and \ref{sec:M-matrix-Quadratic-Optimization} require M-matrix structure. 

For both \texttt{MS} and \texttt{QO}, our IPM is based on defining
a barrier objective parametrized by a scalar $\mu>0$, called the
\textit{centrality parameter}
\[
G_{\mu}\left(x\right)=\frac{1}{\mu}\left(\frac{1}{2}x^{\top}Ax-b^{\top}x\right)-\sum_{i=1}^{n}\ln x_{i}\,.
\]
As is standard in IPMs, the barrier function, which consists of the
logarithmic terms, aims to keep us far away from the boundary of the
non-negative orthant. The centrality parameter $\mu$ balances between
the importance of the barrier and that of the quadratic objective.

For each $\mu>0$, we say that the minimizer $x_{\mu}$ of $G_{\mu}$
is a $\mu$-central point. Notably the central point satisfies the
first order optimality condition:
\[
\nabla G_{\mu}\left(x\right)=\frac{1}{\mu}\left(Ax-b\right)-\frac{1}{x}=0\,.
\]
The set of $\mu$-central points for $\mu>0$ defines the \textit{central
path}. Our goal will be to find our target solution by navigating
the central path. Namely, starting from an initial $\mu$-central
point, we compute another $\mu'$-central point ($\mu'>\mu$ for \texttt{MS},
and $\mu'<\mu$ for \texttt{QO). }We do so via a sequence of Newton
steps that attempt to reestablish $\mu'$-centrality for our current
point $x_{\mu}$. These can be interpreted as a predictor step (which
essentially moves tangent to the central path and improves the ``distance
to centrality'' under an appropriate notion of distance), followed
by a sequence of corrector steps that attempt to reestablish centrality.\footnote{This can be achieved in very few iterations, up to machine precision.
Exact centrality can also be ensured by slightly perturbing the barrier
function to achieve exact gradient optimality. However, as these adjustments
are purely technical and do not provide additional insights, we will
omit further discussion on this point.} This formulation is convenient, since it is a standard fact that
the predictor step dominates the update, as the remaining corrector
steps only handle the higher order terms in the error induced by lack
of centrality.

The main challenge, then, is to ensure that, starting from a $\mu$-central
point $x_{\mu}$, we can consistently take a long predictor step to
a point $x'$, which can subsequently be adjusted via corrector steps
to reach a $\mu'$-central point for $\mu'=\mu\left(1\pm\delta\right)$.
Our primary measure of convergence is the number of such iterations
required to change the centrality parameter $\mu$ by a factor of
2. As we will see later in the main correction lemma (Lemma \ref{lem:correction-lemma}),
to be able to restore $\mu'$-centrality after taking a predictor
step from $x$ it suffices that we set 
\[
\delta=\frac{1}{32\left\Vert \congsimple\right\Vert _{3}}\,,
\]
where $\congsimple$ is the congestion vector defined as:
\[
\congsimple=\left(\frac{1}{\mu'}XAX+I\right)^{-1}1\,.
\]
Hence the $\ell_{3}$ norm of this vector determines the length of
the step we can take.\footnote{The reason we call $\congsimple$ a congesion vector is that it determines
the multiplicative update on $x$, via $x'=x\left(1\pm\delta\congsimple\right)$.
To guarantee strict feasibility for the new iterate, one must keep
the $\ell_{\infty}$ magnitude of the term $\delta\congsimple$ in
the multiplicative update below $1$, so we must control how much
the update $\delta\congsimple$ ``congests'' the old iterate. When
studying network flow problems, this quantity exactly determines how
much an augmenting flow computed by the IPM congests the residual
capacities in a certain symmetrized graph.} 

Unfortunately, it is generally impossible to provide an upper bound
on $\left\Vert \congsimple\right\Vert _{3}$ better than $\sqrt{n}$.
In fact, it is easy to construct instances that force a short step,
ruling out the possibility of improving iteration complexity by guaranteeing
consistently long steps. This represents a major limitation in the
standard analysis of IPMs, as self-concordance theory examines the
problem only locally, aiming to bound the progress made at each individual
step. This step-by-step approach, however, inherently fails to yield
any substantial improvements in overall iteration complexity.

Instead, we will resort to an amortized analysis, which shows that,
while short steps can happen, they do not happen often. Our approach
is strongly inspired by the work in \cite{madry2013navigating,madry2016computing,cmsv17},
where the authors applied an amortized analysis to their IPM. In their
case, the potential function was defined as the energy of a specific
electrical flow computed during the predictor step, which, upon closer
examination, corresponds to the dual local norm of the barrier gradient
at a point on the central path. Following their idea,  it is tempting
to define our potential function as 
\[
\Phi_{\mu}=1^{\top}\left(\frac{1}{\mu}X_{\mu}AX_{\mu}+I\right)^{-1}1\,.
\]
This definition is useful because we can relate its instantaneous
change as we move along the central path to the the $\ell_{3}$ norm
of the congestion vector $\left\Vert \congsimple\right\Vert _{3}$.
Formally, one can calculate 
\begin{align*}
\nabla_{x}\left(1^{\top}\left(\frac{1}{\mu}XAX+I\right)^{-1}1\right) & =-2\left(\left(\frac{1}{\mu}XAX+I\right)^{-1}1\right)\odot AX\left(\frac{1}{\mu}XAX+I\right)^{-1}1\\
 & =-\frac{2}{\mu}\diag{\congsimple}AX\congsimple
\end{align*}
and thus see that 
\begin{align*}
 & \lim_{\delta\rightarrow0}\frac{1}{\delta}\left(1^{\top}\left(\frac{1}{\mu}\diag{1+\delta\congsimple}XAX\diag{1+\delta\congsimple}+I\right)^{-1}1-1^{\top}\left(\frac{1}{\mu}X_{\mu}AX_{\mu}+I\right)^{-1}1\right)\\
= & \left\langle -\frac{2}{\mu}\diag{\congsimple}AX\congsimple,X\congsimple\right\rangle =-\frac{2}{\mu}\left(\congsimple^{2}\right)^{\top}XAX\congsimple=-2\left(\congsimple^{2}\right)^{\top}\left(\frac{1}{\mu}XAX+I\right)\congsimple+2\left\langle \congsimple^{3},1\right\rangle \\
= & -2\left\langle \congsimple^{2},1\right\rangle +2\left\langle \congsimple^{3},1\right\rangle \,.
\end{align*}
This means that (neglecting the change in $\mu$), as we take a tiny
step along the central path, the change in potential function is directly
related to the sum of cubes of $\congsimple$. The difficulty with
this in general is that it is impossible to relate it to $\left\Vert \congsimple\right\Vert _{3}$,
since the former depends on the signs of elements of $\congsimple$.
Assuming that $\left\langle \congsimple^{3},1\right\rangle \approx\left\Vert \congsimple\right\Vert _{3}^{3}$,
we would be in great shape, as it would allow us to directly relate
the change in potential to the length of the predictor step taken. 

Taking this a step further, imagine in addition that this formula
for the change in potential (approximately) holds for longer steps,
where we move between central points $x_{\mu}$ and $x_{\frac{\mu}{1-\delta}}$
where $\delta\approx\frac{1}{2\left\Vert \congsimple\right\Vert _{3}}$,
as in the case of matrix scaling. Then one would have that
\begin{align*}
\Phi_{\frac{\mu}{1-\delta}}-\Phi_{\mu} & \approx-2\delta\left\Vert \congsimple\right\Vert _{2}^{2}+2\delta\left\Vert \congsimple\right\Vert _{3}^{3}\apprge-n^{1/2+1/6}+\left\Vert \congsimple\right\Vert _{3}^{2}\,,
\end{align*}
where we use the fact that $\left\Vert \congsimple\right\Vert _{2}\leq\sqrt{n}$,
as well as the fact that $\left\Vert \congsimple\right\Vert _{2}\leq\left\Vert \congsimple\right\Vert _{3}n^{1/6}.$
In this case we could easily show that steps with large $\left\Vert \congsimple\right\Vert _{3}$
do not occur too often. 

The reason is that we can break down the iterations executed on a
portion of the central path between $\mu_{0}$- and $2\mu_{0}$-centrality,
which we call a phase, into
\begin{enumerate}
\item \textit{short-step iterations} where $\left\Vert \congsimple\right\Vert _{3}\geq2n^{1/3}$,
in which case the potential function must increase by $\Omega\left(n^{2/3}\right)$,
and 
\item \textit{long-step iterations} where $\left\Vert \congsimple\right\Vert _{3}\leq2n^{1/3}$;
in this latter case the potential function can decrease, but by at
most $n^{1-1/3}$.
\end{enumerate}
Since there can not be more than $O\left(n^{1/3}\right)$ long-step
iterations the total decrease in potential caused by these is $O\left(n\right)$.
Noting that $\Phi_{\mu}$ is also bounded from above by $n$, we conclude
that the total increase in potential caused by the short-step iterations
is $O\left(n\right)$. Therefore there can be at most $O\left(n^{1/3}\right)$
short-step iterations, which show that this phase requires $O\left(n^{1/3}\right)$
iterations.

However, executing this argument poses two fundamental challenges.
\begin{enumerate}
\item We assumed that $\left\langle \congsimple^{3},1\right\rangle \approx\left\Vert \congsimple\right\Vert _{3}^{3}$,
which is not true in general.
\item Our argument which measures the change in potential holds for very
short steps.
\end{enumerate}
Our main contribution is to show that these challenges can be overcome
in the special case where the input matrix $A$ is an M-matrix. One
fundamental reason is that M-matrices have a \textit{pointwise non-negative
inverse}. Since one can show that all the matrices encountered, which
take the form $\frac{1}{\mu}XAX+I$ are also M-matrices, we obtain
the key property that $\congsimple\geq0$. Therefore the first property
is guaranteed to hold.

The more technical part of our argument involves showing that the
change in potential is in some sense well approximated by $-\delta\left\Vert \congsimple\right\Vert _{2}^{2}+\delta\left\Vert \congsimple\right\Vert _{3}^{3}$.
We can prove this, up to constant factors involving the two terms,
by using some elementary matrix analysis together with the properties
of M-matrices. Crucially, we rely on a stability lemma (Lemma \ref{lem:rho-mu-stability})
which upper and lower bounds the congestion vector when computed with
respect to different target centrality parameters. This will yield
Lemma \ref{lem:energy-lemma-scaling} and Lemma \ref{lem:energy-lemma-scaling-backward},
which we will use for an amortized analysis nearly identical to the
one we have just described.

A final challenge lies in defining the potential function, as it seems
that the scalar that scales $X_{\mu}AX_{\mu}$ would need to change
at each iteration, complicating matters. To address this, we define
a new potential function
\[
\Phi_{\mu}=1^{\top}\left(\frac{1}{\mu_{0}}X_{\mu}AX_{\mu}+I\right)^{-1}1\,,
\]
for each phase in which we move along the central path between $\mu_{0}$
and $2\mu_{0}$-central points. This keeps the scaling fixed, allowing
only $X_{\mu}$ to evolve monotonically. By leveraging the monotonicity
properties of $\congsimple$, we can show that this potential function
fits our analysis.

\section{Preliminaries}

Throughout this paper, standard arithmetic operators applied to pairs
of vectors should be interpreted pointwise when clear from the context.
We denote vectors in lowercase, with their uppercase counterparts
representing diagonal matrices whose entries match those of the given
vector (e.g., for a vector $x$, we write $X=\diag x$). Inner products
are denoted by $\left\langle \cdot,\cdot\right\rangle $. For a matrix
$A$, we denote its spectral radius by $\rho\left(A\right)$; since
we work with symmetric matrices, we have $\rho\left(A\right)=\left\Vert A\right\Vert $.

\subsection{M-matrices\label{subsec:m-matrix-prop}}

Throughout our analysis we crucially use a few facts about M-matrices,
exhaustively documented in the numerical analysis literature \cite{johnson1982inverse,windisch2013m}.
\begin{fact}
If $A$ is a symmetric M-matrix, then
\begin{enumerate}
\item There is a scalar $s>0$ and an entrywise non-negative matrix $C$
such that $\rho\left(C\right)<s$ and $A=sI-C$. Consequently its
inverse can be written as a convergent series $A^{-1}=\frac{1}{s}\sum_{k\geq0}\left(\frac{1}{s}C\right)^{k},$
which shows that $A^{-1}$ is pointwise non-negative.
\item For any positive diagonal matrix $D$ and any scalar $\alpha\geq0$,
$DAD+\alpha I$ is an M-matrix.
\item There exists a positive diagonal matrix $D$ such that $DAD$ is diagonally
dominant.
\end{enumerate}
\end{fact}

Using these properties we will derive our desired statements by crucially
leveraging pointwise inequalities between non-negative vectors.

\section{Scaling M-matrices\label{sec:scaling}}

Given a symmetric positive semidefinite matrix $A\in\mathbb{R}^{n\times n}$,
we seek to find a vector $x>0$ such that $Ax>0$. Solving this problem
in full generality solves the conic feasibility problem which is equivalent
to linear programming. Here, however, we will provide a faster algorithm
for the case where $A$ is an M-matrix. We first provide the vanilla
analysis based on standard interior point methods, then show how the
same analysis can be refined to achieve a better iteration complexity
in the case of M-matrices.

\subsection{Setting up the IPM}

To solve this problem, we equivalently consider the scaling problem
of finding $x>0$ such that $XAX1=1$, as per Khachiyan-Kalantari
\cite{khachiyan1992diagonal}. This can also be seen as finding a
zero of $Ax-\frac{1}{x}$. To do so via an IPM we define the function
\[
f\left(x\right)=\frac{1}{2}x^{\top}Ax-b^{\top}x\,,
\]
where
\[
b=A1-1\,,
\]
and define the barrier objective
\begin{equation}
G_{\mu}\left(x\right)=\frac{1}{\mu}f\left(x\right)-\sum\ln x\,,\label{eq:barrier-obj}
\end{equation}
which satisfies 
\begin{align*}
\nabla G_{\mu}\left(x\right) & =\frac{1}{\mu}\nabla f\left(x\right)-\frac{1}{x}\\
\nabla^{2}G_{\mu}\left(x\right) & =\frac{1}{\mu}\nabla^{2}f\left(x\right)+X^{-2}\,.
\end{align*}
We note that by definition
\[
\nabla G_{1}\left(1\right)=0\,,
\]
so $x_{0}=1$ is $1$-central. Next we show that achieving gradient
optimality for a large value of $\mu$ provides an (approximate) solution
to our scaling problem.
\begin{lem}
\label{lem:termination-mu}Let $x$ such that $\nabla G_{\mu}\left(x\right)=0$
for $\mu\geq\frac{\left\Vert b\right\Vert _{2}^{2}}{\varepsilon^{2}}$.
Then letting $\widetilde{x}=\frac{x}{\sqrt{\mu}}$ one has that 
\[
\left\Vert \widetilde{X}A\widetilde{X}1-1\right\Vert _{2}\leq\varepsilon\,.
\]
\end{lem}

\begin{proof}
By the gradient optimality condition we have 
\[
Ax-b=\frac{\mu}{x}\,,
\]
and equivalently
\begin{align*}
A\frac{x}{\sqrt{\mu}}-\frac{b}{\sqrt{\mu}} & =\frac{1}{x/\sqrt{\mu}}\\
A\widetilde{x}-\frac{b}{\sqrt{\mu}} & =\frac{1}{\widetilde{x}}\\
\widetilde{X}A\widetilde{X}1 & =1+\frac{b}{\sqrt{\mu}}\,,
\end{align*}
Therefore by rearranging terms and taking norms from both sides we
obtain
\[
\left\Vert \widetilde{X}A\widetilde{X}1-1\right\Vert _{2}=\frac{\left\Vert b\right\Vert _{2}}{\sqrt{\mu}}\,,
\]
which concludes the proof.
\end{proof}
Hence our approach will be to follow the central path by slowly increasing
$\mu$ until it becomes sufficiently large that the contribution of
the $b$ vector gets canceled. In what follows we will provide a vanilla
analysis of the IPM, for general symmetric positive semidefinite matrices
$A$.

\subsection{Following the central path, slowly}

We consider a standard predictor-corrector method. The key primitive
is the correction step which improves centrality, provided we are
not very far from the central path.

\subsubsection{Correction step}

Given the barrier objective $G_{\mu}$ from (\ref{eq:barrier-obj})
we consider the Newton step
\begin{align*}
x' & =x-\left(\nabla^{2}G_{\mu}\left(x\right)\right)^{-1}\nabla G_{\mu}\left(x\right)\\
 & =x\left(1-X^{-1}\left(\nabla^{2}G_{\mu}\left(x\right)\right)^{-1}\nabla G_{\mu}\left(x\right)\right)\,.
\end{align*}
From now on it will be convenient to define the \textit{congestion
vector.}
\begin{defn}
[congestion vector] Given an iterate $x$ and a centrality parameter
$\mu$, we define the congestion vector as
\begin{align}
\cong x{\mu} & =-X^{-1}\left(\nabla^{2}G_{\mu}\left(x\right)\right)^{-1}\nabla G_{\mu}\left(x\right)\,.\label{eq:rho-def}
\end{align}
\end{defn}

This definition captures the multiplicative update performed on $x$
during a Newton step. Based on this we can explicitly write down the
value of $\nabla G_{\mu}\left(x'\right)$ as
\begin{align}
\nabla G_{\mu}\left(x'\right) & =\nabla G_{\mu}\left(x\right)+\nabla G_{\mu}\left(x'\right)-\nabla G_{\mu}\left(x\right)\nonumber \\
 & =-\nabla^{2}G_{\mu}\left(x\right)X\cong x{\mu}+\nabla G_{\mu}\left(x'\right)-\nabla G_{\mu}\left(x\right)\tag*{(using (\ref{eq:rho-def}))}\nonumber \\
 & =-\nabla^{2}G_{\mu}\left(x\right)X\cong x{\mu}+\nabla^{2}G_{\mu}\left(x\right)\left(x'-x\right)+\left(\nabla G_{\mu}\left(x'\right)-\nabla G_{\mu}\left(x\right)-\nabla^{2}G_{\mu}\left(x\right)\left(x'-x\right)\right)\nonumber \\
 & =\nabla G_{\mu}\left(x'\right)-\nabla G_{\mu}\left(x\right)-\nabla^{2}G_{\mu}\left(x\right)\left(x'-x\right)\nonumber \\
 & =-\frac{1}{x'}+\frac{1}{x}-\frac{x'-x}{x^{2}}\nonumber \\
 & =\left(x'-x\right)\left(\frac{1}{xx'}-\frac{1}{x^{2}}\right)=\left(x'-x\right)\left(\frac{x}{x^{2}x'}-\frac{x'}{x^{2}x'}\right)\nonumber \\
 & =\left(x'-x\right)^{2}\cdot\frac{1}{x^{2}x'}=-\frac{\left(X\cong x{\mu}\right)^{2}}{x^{2}x'}\nonumber \\
 & =-\frac{\cong x{\mu}^{2}}{x'}\,.\label{eq:new-grad-correct}
\end{align}
This identity enables us to prove the following lemma:
\begin{lem}
[correction lemma]\label{lem:correction-lemma}One has that 
\[
\left\Vert \nabla G_{\mu}\left(x\right)\right\Vert _{\left(\nabla^{2}G_{\mu}\left(x\right)\right)^{-1}}\leq\left\Vert \cong x{\mu}\right\Vert _{2}\,.
\]
In addition, after a correction step $x'=x\left(1+\cong x{\mu}\right)$
one has that
\[
\left\Vert \cong{x'}{\mu}\right\Vert _{2}\leq\left\Vert \cong x{\mu}\right\Vert _{4}^{2}\,.
\]
\end{lem}

\begin{proof}
For the first part we write
\begin{align*}
\left\Vert \nabla G_{\mu}\left(x\right)\right\Vert _{\left(\nabla^{2}G_{\mu}\left(x\right)\right)^{-1}} & =\left\Vert \left(\nabla^{2}G_{\mu}\left(x\right)\right)^{-1}\nabla G_{\mu}\left(x\right)\right\Vert _{\left(\nabla^{2}G_{\mu}\left(x\right)\right)}\\
 & \leq\left\Vert \left(\nabla^{2}G_{\mu}\left(x\right)\right)^{-1}\nabla G_{\mu}\left(x\right)\right\Vert _{X^{-2}}\\
 & =\left\Vert X^{-1}\left(\nabla^{2}G_{\mu}\left(x\right)\right)^{-1}\nabla G_{\mu}\left(x\right)\right\Vert _{2}\\
 & =\left\Vert \cong x{\mu}\right\Vert _{2}\,.
\end{align*}
Based on (\ref{eq:new-grad-correct}) we can write
\begin{align*}
\left\Vert \cong{x'}{\mu}\right\Vert _{2} & =\left\Vert -X'{}^{-1}\left(\nabla^{2}G_{\mu}\left(x'\right)\right)^{-1}\nabla G_{\mu}\left(x'\right)\right\Vert _{2}\\
 & =\left\Vert \nabla G_{\mu}\left(x'\right)\right\Vert _{\left(\nabla^{2}G_{\mu}\left(x'\right)\right)^{-1}X'{}^{-2}\left(\nabla^{2}G_{\mu}\left(x'\right)\right)^{-1}}\\
 & \leq\left\Vert \nabla G_{\mu}\left(x'\right)\right\Vert _{\left(\nabla^{2}G_{\mu}\left(x'\right)\right)^{-1}}\\
 & \leq\left\Vert \nabla G_{\mu}\left(x'\right)\right\Vert _{X'{}^{2}}\\
 & =\left\Vert -\frac{\cong x{\mu}^{2}}{x'}\right\Vert _{X'{}^{2}}\\
 & =\left\Vert \cong x{\mu}\right\Vert _{4}^{2}\,.
\end{align*}
\end{proof}
Based on Lemma \ref{lem:correction-lemma} we easily see that as soon
as 
\begin{equation}
\left\Vert \cong x{\mu}\right\Vert _{4}\leq\frac{1}{2},\label{eq:condition-for-correctability}
\end{equation}
correction steps decrease $\left\Vert \cong x{\mu}\right\Vert _{2}$
very fast. This will allow us to determine the step size we use when
walking along the central path.
\begin{lem}
\label{lem:rho-def-bd}Let $x_{\mu}$ be $\mu$-central in the sense
that $\nabla G_{\mu}\left(x_{\mu}\right)=0$. Let $0<\delta<1$ and
$\mu'=\frac{\mu}{1-\delta}$. Then $\cong{x_{\mu}}{\mu'}$ satisfies
the identity
\[
\cong{x_{\mu}}{\mu'}=\delta\cdot\left(\frac{1}{\mu'}X_{\mu}AX_{\mu}+I\right)^{-1}1
\]
and 
\[
\left\Vert \cong{x_{\mu}}{\mu'}\right\Vert _{2}\leq\delta\sqrt{n}\,.
\]
\end{lem}

\begin{proof}
Given that we have gradient optimality for $G_{\mu}$, we can express
the gradient of $G_{\mu'}$ as
\begin{align*}
\nabla G_{\mu'}\left(x_{\mu}\right) & =\left(1-\delta\right)\nabla G_{\mu}\left(x_{\mu}\right)-\delta\frac{1}{x_{\mu}}\\
 & =-\delta\frac{1}{x_{\mu}}\,.
\end{align*}
Therefore
\begin{align*}
\cong{x_{\mu}}{\mu'} & =-X_{\mu}^{-1}\left(\nabla^{2}G_{\mu'}\left(x_{\mu}\right)\right)^{-1}\nabla G_{\mu'}\left(x_{\mu}\right)\\
 & =-X_{\mu}^{-1}\left(\frac{1}{\mu'}A+X_{\mu}^{-2}\right)^{-1}\cdot\left(-\delta X_{\mu}^{-1}1\right)\\
 & =\delta\left(\frac{1}{\mu'}X_{\mu}AX_{\mu}+I\right)^{-1}1\,.
\end{align*}
Additionally this allows us to bound 
\begin{align*}
\left\Vert \cong{x_{\mu}}{\mu'}\right\Vert _{2} & =\left\Vert \delta\left(\frac{1-\delta}{\mu}X_{\mu}AX_{\mu}+I\right)^{-1}1\right\Vert _{2}\\
 & \leq\delta\left\Vert 1\right\Vert _{2}=\delta\sqrt{n}\,,
\end{align*}
where we used the fact that 
\[
\left(\frac{1-\delta}{\mu}X_{\mu}AX_{\mu}+I\right)^{-1}\preceq I\,.
\]
\end{proof}
This shows that from a central point we can dial up $\mu$ to $\mu'=\mu/\left(1-\frac{1}{2\sqrt{n}}\right)$
and ensure that the current iterate can be corrected to a $\mu'$-central
point. Together with Lemma \ref{lem:termination-mu}, this allows
us to establish a bound of
\[
O\left(\sqrt{n}\log\frac{\left\Vert A1-1\right\Vert _{2}}{\varepsilon}\right)
\]
on the number of predictor-corrector iterations for the interior point
method. 

In the next section we will established an improved iteration complexity,
assuming that $A$ is an M-matrix. We crucially rely on the observation
that, once we are at a $\mu$-central point, and decide how long of
a step we can afford to take, we only need to ensure that the remaining
error is correctable. According to Lemma \ref{lem:correction-lemma},
this is determined by the $\ell_{4}$ norm of the predictor step,
and hence we can move to $\mu'=\frac{\mu}{1-\delta}$, where $\delta$
is the largest scalar such that
\[
\left\Vert \cong{x_{\mu}}{\mu'}\right\Vert _{4}=\left\Vert \delta\left(\frac{1}{\mu'}X_{\mu}AX_{\mu}+I\right)^{-1}1\right\Vert _{4}\leq\frac{1}{4}\,.
\]
While in general the $\ell_{4}$ norm can be as large as the $\ell_{2}$
norm, in which case we make as much progress as in a standard short-step
IPM, we can show via a potential function analysis that such events
are rare. In fact our criterion for the lengh of the step taken is
slightly weaker, since it is merely determined by the $\ell_{3}$
norm $\left\Vert \cong{x_{\mu}}{\mu'}\right\Vert _{3}$, but it turns
out that this choice works better with our potential function. Ultimately,
this step size will determine the improved iteration complexity of
$O\left(n^{1/3}\log\frac{\left\Vert A1-1\right\Vert _{2}}{\varepsilon}\right)$.

\subsection{Improved analysis using a potential function\label{subsec:Improved-analysis-potential-fn}}

In this section we will focus on the number of iterations require
to advance from a $\mu_{0}$-central point $x_{\mu_{0}}$ to a $\left(2\mu_{0}\right)$-central
point $x_{2\mu_{0}}$. We call this subsequence of iterations a $\mu_{0}$-phase.
To show that a phase requires $O\left(n^{1/3}\right)$ iterations
we will introduce a potential function, which we analyze using the
properties of M-matrices.

Since our analysis hinges on quantities involving the Hessian matrix
applied to the all-ones vector, it is convenient to introduce a new
definition.
\begin{defn}
[unnormalized congestion vector] Given a vector $x$ and a scalar
$\mu$, we define the unnormalized congestion
\begin{equation}
\nocong x{\mu}=\left(\frac{1}{\mu}XAX+I\right)^{-1}1\,.\label{eq:unnorm-cong}
\end{equation}
\end{defn}

Given this definition, we highlight the relation between the congestion
vector $\rho$ and unnormalized congestion vector $\widehat{\rho}$
at central points. 
\begin{cor}
Given a $\mu$-central point $x_{\mu}$ one has that 
\begin{align*}
\cong{x_{\mu}}{\frac{\mu}{1-\delta}} & =\delta\nocong{x_{\mu}}{\frac{\mu}{1-\delta}}\,.
\end{align*}
\end{cor}

Next, for a $\mu_{0}$-phase we define the potential function
\begin{equation}
\Phi_{\mu}:=1^{\top}\left(\frac{1}{\mu_{0}}X_{\mu}AX_{\mu}+I\right)^{-1}1\,,\label{eq:potential-fn}
\end{equation}
where $x_{\mu}$ is a $\mu$-central point for $\mu\in\left[\mu_{0},2\mu_{0}\right]$.
Critically, we use the fact that since $A$ is an M-matrix, the underlying
Newton steps force $x$ to be monotone in the sense that all of its
coordinates can only increase. Intuitively (at least in the case when
$A$ is diagonal), this should increase $\Phi_{\mu}$; while this
is not generally true since $X_{\mu}$ and $A$ do not commute, we
can show that our desired potential function monotonicity ``almost''
holds. The key of this section will be to show that the increase in
$\Phi_{\mu}$ after each predictor-corrector sequence executed in
a phase is determined by the norm of $\nocong{x_{\mu}}{\frac{\mu}{1-\delta}}$.
Formally we will show that, roughly, 
\[
\Phi_{\frac{\mu}{1-\delta}}-\Phi_{\mu}\apprge-\delta\left\Vert \nocong{x_{\mu}}{\frac{\mu}{1-\delta}}\right\Vert _{2}^{2}+\delta\left\Vert \nocong{x_{\mu}}{\frac{\mu}{1-\delta}}\right\Vert _{3}^{3}\,.
\]
This ensures that when taking steps of length $\delta\approx\left\Vert \nocong{x_{\mu}}{\frac{\mu}{1-\delta}}\right\Vert _{3}$
along the central path, the change in $\Phi$ is dominated by $\delta\left\Vert \nocong{x_{\mu}}{\frac{\mu}{1-\delta}}\right\Vert _{3}^{3}$.
Hence in the situation when $\delta$ is small, because the $\ell_{3}$
norm of the normalized congestion vector is too large, the potential
function increases by a lot. Together with the fact that 
\begin{align}
\Phi_{\mu} & \leq1^{\top}\left(I\right)^{-1}1\leq n\,,\label{eq:phi-ub}
\end{align}
we will conclude that there can not possibly be too many short steps.

To proceed with the formal analysis, we first establish a series of
helper lemmas.
\begin{lem}
[stability of the normalized congestion under changes in $\mu$]\label{lem:rho-mu-stability}
If $A$ is an M-matrix and $\mu_{0}<\mu$ are scalars then 
\[
\nocong x{\mu_{0}}\geq\frac{\mu_{0}}{\mu}\nocong x{\mu}\,,
\]
pointwise. Additionally, 
\begin{align*}
\left\Vert \nocong x{\mu}\right\Vert _{2} & \geq\frac{\mu_{0}}{\mu}\left\Vert \nocong x{\mu_{0}}\right\Vert _{2}\,.
\end{align*}
\end{lem}

\begin{proof}
Since $A$ is an M-matrix then $B=\frac{1}{\mu_{0}}XAX$ is also an
M-matrix, and hence one can write $B=sI-C$, where $\rho\left(C\right)<s$
and $C\geq0$ pointwise. Therefore
\begin{align*}
\nocong x{\mu_{0}} & =\left(\frac{1}{\mu_{0}}XAX+I\right)^{-1}1=\left(B+I\right)^{-1}1=\left(sI-C+I\right)^{-1}1\\
 & =\left(\left(s+1\right)I-C\right)^{-1}1=\frac{1}{s+1}\left(I-\frac{1}{s+1}C\right)^{-1}1=\frac{1}{s+1}\left(\sum_{k=0}^{\infty}\left(\frac{1}{s+1}C\right)^{k}\right)1\,.
\end{align*}
Similarly we can write
\begin{align*}
\nocong x{\mu} & =\left(\frac{1}{\mu}XAX+I\right)^{-1}1=\left(\frac{\mu_{0}}{\mu}\cdot\frac{1}{\mu_{0}}XAX+I\right)^{-1}1=\left(\frac{\mu_{0}}{\mu}\cdot\left(sI-C\right)+I\right)^{-1}1\\
 & =\frac{\mu}{\mu_{0}}\left(\left(sI-C\right)+\frac{\mu}{\mu_{0}}I\right)^{-1}1=\frac{\mu}{\mu_{0}}\left(\left(s+\frac{\mu}{\mu_{0}}\right)I-C\right)^{-1}1\\
 & =\frac{\mu}{\mu_{0}\left(s+\frac{\mu}{\mu_{0}}\right)}\cdot\left(I-\frac{1}{s+\frac{\mu}{\mu_{0}}}C\right)^{-1}1=\frac{\mu}{\mu_{0}\left(s+\frac{\mu}{\mu_{0}}\right)}\cdot\left(\sum_{k=0}^{\infty}\left(\frac{1}{s+\frac{\mu}{\mu_{0}}}C\right)^{k}\right)1.
\end{align*}
Now we notice that the terms in the power series for $\nocong x{\mu}$
are all smaller than those for $\nocong x{\mu_{0}}$. Therefore we
can conclude that
\[
\nocong x{\mu_{0}}\geq\frac{1}{s+1}\left(\sum_{k=0}^{\infty}\left(\frac{1}{s+\frac{\mu}{\mu_{0}}}C\right)^{k}\right)1=\frac{\mu_{0}\left(s+\frac{\mu}{\mu_{0}}\right)}{\mu\left(s+1\right)}\nocong x{\mu}\geq\frac{\mu_{0}}{\mu}\nocong x{\mu}\,.
\]
For the second part we use the fact that
\[
\frac{1}{\mu}XAX+I\preceq\frac{1}{\mu}XAX+\frac{\mu_{0}}{\mu}I=\frac{\mu_{0}}{\mu}\left(\frac{1}{\mu_{0}}XAX+I\right)
\]
which equivalently gives
\[
\left(\frac{1}{\mu}XAX+I\right)^{-1}\succeq\frac{\mu_{0}}{\mu}\left(\frac{1}{\mu_{0}}XAX+I\right)^{-1}
\]
which allows us to conclude that 
\[
\left\Vert \nocong x{\mu}\right\Vert _{2}\geq\frac{\mu_{0}}{\mu}\left\Vert \nocong x{\mu_{0}}\right\Vert _{2}
\]
which gives us what we need.
\end{proof}

\begin{lem}
[local energy increase]\label{lem:local-energy-increase}If $A$
is a symmetric M-matrix, $x'=x_{\mu}\left(1+\update\right)$, for
some $\update\geq0$, $\mu_{0}\leq\mu$, and $\delta>0$ then 
\begin{align*}
 & 1^{\top}\left(\frac{1}{\mu_{0}}X'AX'+I\right)^{-1}1-1^{\top}\left(\frac{1}{\mu_{0}}X_{\mu}AX_{\mu}+I\right)^{-1}1\\
 & \geq-2\frac{\mu}{\mu_{0}\left(1-\delta\right)}\cdot\left\Vert \update\right\Vert _{2}\left\Vert \nocong{x_{\mu}}{\mu}\right\Vert _{2}+2\left(\frac{1}{1+\left\Vert \update\right\Vert _{\infty}}\right)^{4}\left(\frac{\mu_{0}\left(1-\delta\right)}{\mu}\right)^{2}\left\langle \update,\nocong{x_{\mu}}{\mu}^{2}\right\rangle \,,
\end{align*}
\end{lem}

\begin{proof}
First we apply Lemma $\ref{lem:general-energy-lemma}$ on the matrix
to obtain
\begin{align*}
 & 1^{\top}\left(\frac{1}{\mu_{0}}X'AX'+I\right)^{-1}1-1^{\top}\left(\frac{1}{\mu_{0}}X_{\mu}AX_{\mu}+I\right)^{-1}1\\
 & \geq-2\left\langle \update,\nocong{x_{\mu}}{\mu_{0}}\right\rangle +2\left(\frac{1}{1+\left\Vert \update\right\Vert _{\infty}}\right)^{4}\left\langle \update,\nocong{x_{\mu}}{\mu_{0}}^{2}\right\rangle \\
 & \geq-2\left\Vert \update\right\Vert _{2}\left\Vert \nocong{x_{\mu}}{\mu_{0}}\right\Vert _{2}+2\left(\frac{1}{1+\left\Vert \update\right\Vert _{\infty}}\right)^{4}\left\langle \update,\nocong{x_{\mu}}{\mu_{0}}^{2}\right\rangle \,.
\end{align*}
Applying Lemma \ref{lem:rho-mu-stability} we can continue this chain
of inequalities with 
\[
\geq-2\frac{\mu}{\mu_{0}\left(1-\delta\right)}\left\Vert \update\right\Vert _{2}\left\Vert \nocong{x_{\mu}}{\frac{\mu}{1-\delta}}\right\Vert _{2}+2\left(\frac{\mu_{0}\left(1-\delta\right)}{\mu}\right)^{2}\left(\frac{1}{1+\left\Vert \update\right\Vert _{\infty}}\right)^{4}\left\langle \update,\nocong{x_{\mu}}{\frac{\mu}{1-\delta}}^{2}\right\rangle \,.
\]
This concludes the proof.
\end{proof}
We now use Lemma \ref{lem:local-energy-increase} to lower bound the
change in energy after moving between central points.
\begin{lem}
[energy lemma]\label{lem:energy-lemma-scaling}Let $x_{\mu}$ and
$x_{\mu'}$ be points on the central path, for $\mu'=\frac{\mu}{1-\delta}$,
$1/2>\delta>0$ such that $\left\Vert \delta\nocong{x_{\mu}}{\frac{\mu}{1-\delta}}\right\Vert _{3}=\alpha\leq\frac{1}{16}$.
Then
\[
\Phi_{\mu'}\geq\Phi_{\mu}-2^{3}\alpha n^{1/6}\left\Vert \nocong{x_{\mu}}{\frac{\mu}{1-\delta}}\right\Vert _{2}+2^{-11}\alpha\left\Vert \nocong{x_{\mu}}{\frac{\mu}{1-\delta}}\right\Vert _{3}^{2}\,.
\]
\end{lem}

\begin{proof}
To establish this lower bound on the increase in potential, we first
use Lemma \ref{lem:full-correction-update} which shows that the multiplicative
change in $x_{\mu}$ determined by the Newton predictor step $x'=x_{\mu}\left(1+\delta\nocong{x_{\mu}}{\frac{\mu}{1-\delta}}\right)$
approximates very well the multiplicative change encountered when
moving straight to $x_{\frac{\mu}{1-\delta}}$. Indeed, for $\delta$
chosen such that one has that $\left\Vert \delta\nocong{x_{\mu}}{\frac{\mu}{1-\delta}}\right\Vert _{3}=\alpha\leq1/16$,
we have
\[
x_{\frac{\mu}{1-\delta}}=x_{\mu}\left(1+\delta\nocong{x_{\mu}}{\frac{\mu}{1-\delta}}+\vzeta\right)
\]
where $\left\Vert \vzeta\right\Vert _{2}\leq8\alpha^{2}$. Now we
apply Lemma \ref{lem:local-energy-increase}, to lower bound
\begin{align*}
 & 1^{\top}\left(\frac{1}{\mu_{0}}X_{\frac{\mu}{1-\delta}}AX_{\frac{\mu}{1-\delta}}+I\right)^{-1}1-1^{\top}\left(\frac{1}{\mu_{0}}X_{\mu}AX_{\mu}+I\right)^{-1}1\\
 & \geq-2\frac{\mu}{\mu_{0}\left(1-\delta\right)}\cdot\left\Vert \delta\nocong{x_{\mu}}{\frac{\mu}{1-\delta}}+\vzeta\right\Vert _{2}\left\Vert \nocong{x_{\mu}}{\frac{\mu}{1-\delta}}\right\Vert _{2}\\
 & +2\left(\frac{1}{1+\left\Vert \delta\nocong{x_{\mu}}{\frac{\mu}{1-\delta}}+\zeta\right\Vert _{\infty}}\right)^{4}\left(\frac{\mu_{0}\left(1-\delta\right)}{\mu}\right)^{2}\left\langle \delta\nocong{x_{\mu}}{\frac{\mu}{1-\delta}}+\vzeta,\nocong{x_{\mu}}{\frac{\mu}{1-\delta}}^{2}\right\rangle \\
 & \geq-4\left(\left\Vert \delta\nocong{x_{\mu}}{\frac{\mu}{1-\delta}}\right\Vert _{2}+\left\Vert \vzeta\right\Vert _{2}\right)\left\Vert \nocong{x_{\mu}}{\frac{\mu}{1-\delta}}\right\Vert _{2}\\
 & +2\left(\frac{1}{1+\left\Vert \delta\nocong{x_{\mu}}{\frac{\mu}{1-\delta}}\right\Vert _{3}+\left\Vert \vzeta\right\Vert _{2}}\right)^{4}\frac{1}{16}\left(\left\langle \delta\nocong{x_{\mu}}{\frac{\mu}{1-\delta}},\nocong{x_{\mu}}{\frac{\mu}{1-\delta}}^{2}\right\rangle -\left\Vert \vzeta\right\Vert _{2}\left\Vert \nocong{x_{\mu}}{\frac{\mu}{1-\delta}}\right\Vert _{4}^{2}\right)\\
 & \geq-4\delta\left\Vert \nocong{x_{\mu}}{\frac{\mu}{1-\delta}}\right\Vert _{2}^{2}-4\cdot8\alpha^{2}\left\Vert \nocong{x_{\mu}}{\frac{\mu}{1-\delta}}\right\Vert _{2}+2\cdot\frac{1}{16}\cdot\frac{1}{4^{4}}\left(\delta\left\Vert \nocong{x_{\mu}}{\frac{\mu}{1-\delta}}\right\Vert _{3}^{3}-8\alpha^{2}\left\Vert \nocong{x_{\mu}}{\frac{\mu}{1-\delta}}\right\Vert _{3}^{2}\right)\\
 & \geq-4\left(\alpha n^{1/6}+8\alpha^{2}\right)\left\Vert \nocong{x_{\mu}}{\frac{\mu}{1-\delta}}\right\Vert _{2}+2\cdot\frac{1}{16}\cdot\frac{1}{4^{4}}\left(\alpha-8\alpha^{2}\right)\left\Vert \nocong{x_{\mu}}{\frac{\mu}{1-\delta}}\right\Vert _{3}^{2}\\
 & \geq-4\left(2\alpha n^{1/6}\right)\left\Vert \nocong{x_{\mu}}{\frac{\mu}{1-\delta}}\right\Vert _{2}+2\cdot\frac{1}{16}\cdot\frac{1}{4^{4}}\left(\frac{\alpha}{2}\right)\left\Vert \nocong{x_{\mu}}{\frac{\mu}{1-\delta}}\right\Vert _{3}^{2}\tag*{(using \ensuremath{\alpha\leq1/16})}\\
 & =-2^{3}\alpha n^{1/6}\left\Vert \nocong{x_{\mu}}{\frac{\mu}{1-\delta}}\right\Vert _{2}+\frac{1}{2^{11}}\alpha\left\Vert \nocong{x_{\mu}}{\frac{\mu}{1-\delta}}\right\Vert _{3}^{2}\,.
\end{align*}

\end{proof}
Using Lemma \ref{lem:energy-lemma-scaling} we can now prove the improved
iteration bound.
\begin{lem}
Given a central point $x_{\mu_{0}}$ , one can find a new central
point $x_{\mu}$, $\mu\geq2\mu_{0}$, in $O\left(n^{1/3}\right)$
iterations, each of which involves a single predictor step, following
by recentering via $\widetilde{O}\left(1\right)$ correction steps.
\end{lem}

\begin{proof}
Our predictor step consists of computing $\nocong{x_{\mu}}{\mu}$,
for which by Lemma \ref{lem:rho-mu-stability} we know that $\left\Vert \nocong{x_{\mu}}{\frac{\mu}{1-\delta}}\right\Vert _{3}\leq2\left\Vert \nocong{x_{\mu}}{\mu}\right\Vert _{3}$for
any $0\leq\delta\leq1/2$. Then we set $\delta=\frac{1}{32\left\Vert \nocong{x_{\mu}}{\mu}\right\Vert _{3}}$
which ensures that $\delta\leq\frac{1}{16\left\Vert \nocong{x_{\mu}}{\frac{\mu}{1-\delta}}\right\Vert _{3}}$
and in $\widetilde{O}\left(1\right)$ Newton steps compute $x_{\frac{\mu}{1-\delta}}$.
Note that the bound on $\delta$ suffices for this purposes, per (\ref{eq:condition-for-correctability}).
Naturally, this $\delta$ may be as small as $O\left(n^{-1/2}\right)$
in the worst case, which prevents us from running a worst case analysis.
To amortize the analysis we Indeed, to handle the second case, we
use Lemma \ref{lem:energy-lemma-scaling} which shows that 
\begin{align*}
\Phi_{\frac{\mu}{1-\delta}} & \geq\Phi_{\mu}-2^{3}\alpha n^{1/6}\left\Vert \nocong{x_{\mu}}{\frac{\mu}{1-\delta}}\right\Vert _{2}+2^{-11}\alpha\left\Vert \nocong{x_{\mu}}{\frac{\mu}{1-\delta}}\right\Vert _{3}^{2}\,.
\end{align*}
 where $\alpha=\Theta\left(1\right)$. Now we consider two cases:
\begin{enumerate}
\item $\left\Vert \nocong{x_{\mu}}{\frac{\mu}{1-\delta}}\right\Vert _{3}\leq2^{8}n^{1/3}$,
in which case we perform a long step, as $\frac{\mu}{1-\delta}=\mu\left(1+\Omega\left(n^{-1/3}\right)\right)$.
There can only be $O\left(n^{1/3}\right)$ such steps before doubling
$\mu$. In addition, we verify that 
\[
\Phi_{\frac{\mu}{1-\delta}}\geq\Phi_{\mu}-\Theta\left(n^{1-1/3}\right)
\]
since $\left\Vert \nocong{x_{\mu}}{\frac{\mu}{1-\delta}}\right\Vert _{2}\leq\sqrt{n}$.
\item $\left\Vert \nocong{x_{\mu}}{\frac{\mu}{1-\delta}}\right\Vert _{3}>2^{8}n^{1/3}$,
in which case we must argue that there can not possibly be many such
events. To do so we see that the potential function must increase
very aggresively in this case. Indeed, 
\[
\Phi_{\frac{\mu}{1-\delta}}\geq\Phi_{\mu}+\Theta\left(1\right)\cdot n^{2/3}\,.
\]
\end{enumerate}
Thus the total decrease in $\Phi$ caused by the at most $O\left(n^{1/3}\right)$
events where we take a long step is $O\left(n\right)$. In addition,
each time we are forst to take a short step increases $\Phi$ by $\Omega\left(n^{2/3}\right)$.
Therefore as $\Phi\leq n$ per (\ref{eq:phi-ub}), the total number
of such short step events is bounded by $O\left(n^{1/3}\right)$.
This concludes the proof.
\end{proof}

\section{M-matrix Quadratic Optimization\label{sec:M-matrix-Quadratic-Optimization}}

Here we show that our approach can be used to solve quadratic optimization
problems of the form
\[
\min_{x\geq0}\frac{1}{2}x^{\top}Ax-b^{\top}x\,,
\]
which capture the $\ell_{2}$ diffusion problem as shown in \cite{chen20222}.
While their algorithm readily provides a nearly linear time running
time, it is unclear to what extent it parallelizes, and it is exclusively
based on leveraging the combinatorial structure of graphs. Here we
show that a vanilla long-step IPM, which completely ignores the combinatorial
structure still achieves an improved running time beyond standard
self-concordance guarantees.

We solve this problem similarly to the scaling problem from Section
\ref{sec:scaling}. In fact the barrier function takes exactly the
same form as in the scaling instance. The one major difference, however,
is that previously we were allowed to choose the value of the vector
$b$ in a convenient way, such that we reach a centered instance.
Then, starting from that point, we \textit{increase} the centrality
parameter $\mu$ via Newton steps, while here our goal will be to
get $\mu$ very close to $0$. The difficulty in this case stems from
the fact that it is unclear how to even initialize the method, since
it is unclear how to find a central point. We will show that we can
actually leverage the scaling from Section \ref{sec:scaling} to find
a central starting point. Hence the full IPM consists of two phases.
In the first phase we start with an infeasible instance, which we
will convert into a feasible one but with large centrality parameter.
During the second phase we start with this central point with large
$\mu$ to obtain a central point with very small duality gap.

Just as in the previous section, our barrier formulation will be
\begin{equation}
G_{\mu}\left(x\right)=\frac{1}{\mu}\left(\frac{1}{2}x^{\top}Ax-b^{\top}x\right)-\sum\ln x\,,\label{eq:barrier-diffusion}
\end{equation}
and, provided we are at a $\mu$-central point we can reach a $\frac{\mu}{1+\delta}$-central
point by correcting the residual. Formally, provided that
\[
\nabla G_{\mu}\left(x_{\mu}\right)=0
\]
we have 
\[
\nabla G_{\frac{\mu}{1+\delta}}\left(x_{\mu}\right)=\frac{1+\delta}{\mu}\left(Ax_{\mu}-b\right)-\frac{1}{x_{\mu}}=\left(1+\delta\right)\left(\frac{1}{\mu}\left(Ax_{\mu}-b\right)-\frac{1}{x_{\mu}}\right)+\delta\cdot\frac{1}{x_{\mu}}=\delta\cdot\frac{1}{x_{\mu}}\,.
\]
Therefore the predictor step takes the form
\begin{align*}
x' & =x_{\mu}-\left(\nabla^{2}G_{\frac{\mu}{1+\delta}}\left(x_{\mu}\right)\right)^{-1}\nabla G_{\frac{\mu}{1+\delta}}\left(x_{\mu}\right)\\
 & =x_{\mu}-\left(\frac{1+\delta}{\mu}A+X_{\mu}^{-2}\right)^{-1}\cdot\delta\cdot\frac{1}{x_{\mu}}\\
 & =x_{\mu}\left(1-\delta\cdot X_{\mu}^{-1}\left(\frac{1+\delta}{\mu}A+X_{\mu}^{-2}\right)^{-1}X_{\mu}1\right)\\
 & =x_{\mu}\left(1-\delta\cdot\left(\frac{1+\delta}{\mu}X_{\mu}AX_{\mu}+I\right)^{-1}1\right)\\
 & =x_{\mu}\left(1-\delta\nocong{x_{\mu}}{\frac{\mu}{1+\delta}}\right)\,.
\end{align*}

Next we show how to obtain a central point for the diffusion problem.
\begin{lem}
Let $\left(A,b\right)$ describe an instance of the diffusion problem
$\min_{x\geq0}\frac{1}{2}x^{\top}Ax-b^{\top}x$. Let $\mu=2\left\Vert A1-1-b\right\Vert _{2}$,
and let $x$ be a $\mu$-central point for the scaling problem i.e.
\begin{align*}
\frac{1}{\mu}\left(Ax-b_{0}\right)-\frac{1}{x} & =0\,,\text{ where}\\
b_{0} & =A1-1\,.
\end{align*}
Then $x$ can be corrected to a $\mu$-central point for the diffusion
problem i.e. find a point $x'$ satisfying 
\[
\frac{1}{\mu}\left(Ax'-b\right)-\frac{1}{x'}=0
\]
in $\widetilde{O}\left(1\right)$ iterations of a Newton method.
\end{lem}

\begin{proof}
To prove the Lemma it suffices to show that $\left\Vert \cong x{\mu}\right\Vert _{2}\leq1/2$.
Indeed
\begin{align*}
\left\Vert \cong x{\mu}\right\Vert _{2} & =\left\Vert -X^{-1}\left(\frac{1}{\mu}A+X^{-2}\right)^{-1}\left(\frac{1}{\mu}\left(Ax-b\right)-X^{-1}1\right)\right\Vert _{2}\,.\\
 & =\left\Vert \left(\frac{1}{\mu}XAX+I\right)^{-1}\left(\frac{1}{\mu}\left(Ax-b\right)-X^{-1}1\right)\right\Vert _{2}\\
 & \leq\left\Vert \frac{1}{\mu}\left(Ax-b\right)-X^{-1}1\right\Vert _{2}\\
 & =\left\Vert \frac{1}{\mu}\left(Ax-b_{0}\right)-X^{-1}1+\frac{1}{\mu}\left(b_{0}-b\right)\right\Vert _{2}\,.
\end{align*}
Since $x$ is $\mu$ central for the scaling problem we further have
that
\[
\left\Vert \cong x{\mu}\right\Vert _{2}=\left\Vert \frac{1}{\mu}\left(b_{0}-b\right)\right\Vert _{2}=\frac{1}{\mu}\left\Vert A1-1-b\right\Vert _{2}\leq\frac{1}{2}\,.
\]
By (\ref{eq:condition-for-correctability}) we know that $x$ can
be quickly corrected to centrality for the diffusion problem, which
completes the proof.

\end{proof}
Now our goal will be to reduce $\mu$ sufficiently so that we achieve
a small duality gap. 
\begin{lem}
\label{lem:duality-gap-diffusion}Let $x_{\mu}$ be a $\mu$-central
solution for the diffusion problem. Then if $x^{*}=\frac{1}{2}\min_{x\geq0}x^{\top}Ax-b^{\top}x$,
then $\frac{1}{2}x_{\mu}^{\top}Ax_{\mu}-b^{\top}x_{\mu}\leq\frac{1}{2}x^{*\top}Ax^{*}-b^{\top}x^{*}+\mu n.$
\end{lem}

\begin{proof}
We use a primal-dual argument, since we will use the centrality of
$x_{\mu}$ to bound the duality gap. Consider a Cholesky factorization
of $A=C^{\top}C$ and define the dual problem
\begin{align*}
\min_{z:C^{\top}z\geq b}\frac{1}{2}\left\Vert z\right\Vert _{2}^{2} & =\max_{x\geq0}\min_{z}\frac{1}{2}\left\Vert z\right\Vert _{2}^{2}+\left\langle x,b-C^{\top}z\right\rangle \\
 & =\max_{x\geq0}\left\langle b,x\right\rangle +\min_{z}\frac{1}{2}\left\Vert z\right\Vert _{2}^{2}-\left\langle Cx,z\right\rangle 
\end{align*}
for which we can verify that the inner minimization problem is minimized
at $z=Cx$. Thus we can further write this as 
\[
\max_{x\geq0}\left\langle b,x\right\rangle -\frac{1}{2}\left\Vert Cx\right\Vert _{2}^{2}=\max_{x\geq0}\left\langle b,x\right\rangle -\frac{1}{2}x^{\top}Ax\,.
\]
Given a pair of feasible primal-dual solutions $\left(x,z\right)$
we can now let $s=C^{\top}z-b\geq0$ and write the duality gap as
\begin{align*}
\frac{1}{2}\left\Vert z\right\Vert _{2}^{2}-\left(\left\langle b,x\right\rangle -\frac{1}{2}\left\Vert Cx\right\Vert _{2}^{2}\right) & =\frac{1}{2}\left\Vert z\right\Vert _{2}^{2}-\left(\left\langle C^{\top}z-s,x\right\rangle -\frac{1}{2}\left\Vert Cx\right\Vert _{2}^{2}\right)\\
 & =\frac{1}{2}\left\Vert z\right\Vert _{2}^{2}-\left\langle z,Cx\right\rangle +\frac{1}{2}\left\Vert Cx\right\Vert _{2}^{2}+\left\langle s,x\right\rangle \\
 & =\frac{1}{2}\left\Vert z-Cx\right\Vert _{2}^{2}+\left\langle s,x\right\rangle \,.
\end{align*}
Now since by centrality we have $\frac{1}{\mu}\left(Ax-b\right)-\frac{1}{x}=0$
we can equivalently write $C^{\top}Cx-b=\frac{\mu}{x}>0$, and so
we see that $z=Cx$ is a feasible dual vector. Thus we can measure
the duality gap as
\[
\frac{1}{2}\left\Vert Cx-Cx\right\Vert _{2}^{2}+\left\langle C^{\top}Cx-b,x\right\rangle =\left\langle \frac{\mu}{x},x\right\rangle =\mu n\,.
\]
This implies the desired bound.
\end{proof}
We are now almost ready to prove the improved iteration complexity
of the long step IPM. The final remaining ingredient is a backward
version of Lemma \ref{lem:energy-lemma-scaling}, which lower bounds
$\Phi_{\mu'}-\Phi_{\mu}$ in terms of $\nocong{x_{\mu'}}{\mu'}$ rather
than $\nocong{x_{\mu}}{\mu'}$. This follows from a similar proof.
Since it is technical, we defer it to the appendix as Lemma \ref{lem:energy-lemma-scaling-backward}.

Finally, we can now prove the improved iteration bound.
\begin{lem}
Given a central point $x_{\mu_{0}}$ for the diffusion problem, one
can find a new central point $x_{\mu}$, $\mu\leq\mu_{0}/2$, in $O\left(n^{1/3}\right)$
iterations, each of which involves a single predictor step, following
by recentering via $\widetilde{O}\left(1\right)$ correction steps.
\end{lem}

\begin{proof}
Our predictor step consists of computing via line search a scalar
$1/2\geq\delta>0$ for which $\frac{1}{32}\leq\delta\left\Vert \nocong{x_{\mu}}{\frac{\mu}{1+\delta}}\right\Vert _{3}\leq\frac{1}{16}$.
This is always possible since $f\left(\delta\right)=\delta\left\Vert \nocong{x_{\mu}}{\frac{\mu}{1+\delta}}\right\Vert _{3}$
is a continuous function in $\delta$. Thus we have $\delta\leq\frac{1}{16\left\Vert \nocong{x_{\mu}}{\frac{\mu}{1+\delta}}\right\Vert _{3}}$
and in $\widetilde{O}\left(1\right)$ Newton steps compute the next
central point $x_{\frac{\mu}{1+\delta}}$. Now to perform the amortized
analysis we use Lemma \ref{lem:energy-lemma-scaling-backward} which
shows that 
\[
\Phi_{\frac{\mu}{1+\delta}}\leq\Phi_{\mu}+2^{5}\alpha n^{1/6}\left\Vert \nocong{x_{\mu}}{\frac{\mu}{1+\delta}}\right\Vert _{2}-2^{-1}\alpha\left\Vert \nocong{x_{\mu}}{\frac{\mu}{1+\delta}}\right\Vert _{3}^{2}\,,
\]
where $\alpha=\delta\left\Vert \nocong{x_{\mu}}{\frac{\mu}{1+\delta}}\right\Vert _{3}=\Theta\left(1\right)$.
Again we can consider two cases:
\begin{enumerate}
\item $\left\Vert \nocong{x_{\mu}}{\frac{\mu}{1+\delta}}\right\Vert _{3}\leq2^{8}n^{1/3}$,
in which case we perform a long step, as $\frac{\mu}{1+\delta}=\mu\left(1-\Omega\left(n^{-1/3}\right)\right)$.
There can only be $O\left(n^{1/3}\right)$ such steps before halving
$\mu$. In addition, we verify that 
\[
\Phi_{\frac{\mu}{1+\delta}}\leq\Phi_{\mu}+O\left(n^{1-1/3}\right)
\]
since $\left\Vert \nocong{x_{\mu}}{\frac{\mu}{1+\delta}}\right\Vert _{2}\leq\sqrt{n}$.
\item $\left\Vert \nocong{x_{\mu}}{\frac{\mu}{1+\delta}}\right\Vert _{3}>2^{8}n^{1/3}$,
in which case we must argue that there can not possibly be many such
events. To do so we see that the potential function must increase
very aggresively in this case. Indeed,
\[
\Phi_{\frac{\mu}{1+\delta}}\leq\Phi_{\mu}-\Theta\left(1\right)\cdot n^{2/3}\,.
\]
\end{enumerate}
We use the fact that $\Phi_{\mu}$ is lower bounded by $0$. Additionally,
at the beginning of the phase $\Phi_{\mu}=1^{\top}\left(\frac{1}{\mu_{0}}X_{\mu}AX_{\mu}+I\right)^{-1}1\leq n$,
and over the coures of the phase it can only increase by $O\left(\sqrt{n}\right)$
due to each of the $O\left(n^{1/3}\right)$ events when we take a
long step. Additionally, in each of the events when we are forced
to take a short step, the potential function decreases by at least
$\Theta\left(n^{2/3}\right)$. Hence the total number of such steps
is $O\left(n^{1/3}\right)$.
\end{proof}

\section{Discussion and Outlook}

We have demonstrated that self-concordance theory is not the definitive
framework for understanding the iteration complexity of interior point
methods. This result offers optimism that even tighter bounds on IPM
iteration complexity may be achievable for a broader range of problems.
Several intriguing questions remain:
\begin{itemize}
\item \textbf{Lower bounds for ``reasonable'' IPMs:} Can we construct
an iteration lower bound for ``reasonable'' IPMs on the problems
studied in this paper that matches our $\widetilde{O}\left(n^{1/3}\right)$
upper bound? Alternatively, could the iteration complexity of our
IPMs be further improved beyond this bound?
\item \textbf{Using more general matrices:} Jin and Kalantari \cite{jin2006procedure}
show that matrix scaling for general positive semidefinite matrices
suffices to solve linear programs. However, their reductions produce
matrices that lack the M-matrix structure. Are there broader classes
of matrices where our techniques could apply? If so, could these extensions
enable solutions to a wider class of natural optimization problems?
\item \textbf{Improved iteration complexity for IPMs with decomposable barriers:}
An important feature of our problems is that the barrier function
decomposes, which allows recentering steps within a larger region
than the Dikin ellipsoid. Formally, if the barrier function decomposes
as $\phi\left(x\right)=\sum_{i=1}^{n}\phi_{i}\left(x\right)$, the
corresponding congestion vector can be defined as $\congsimple_{i}^{2}=\left\Vert \left(\nabla^{2}\phi\right)^{-1}\nabla\phi_{i}\left(x\right)\right\Vert _{\nabla^{2}\phi_{i}}^{2}$,
and the ability to recenter is determined by $\left\Vert \rho\right\Vert _{4}$
being smaller than a constant, rather than $\left\Vert \rho\right\Vert _{2}$.
Are there instances where $\left\Vert \congsimple\right\Vert _{4}$
can be more easily controlled?
\item \textbf{Extending our amortized analysis:} A compelling case is that
of linear programs, where the best-known lower bounds, due to Todd
and Ye \cite{todd2020lower,todd1996lower}, already establish an $\Omega\left(n^{1/3}\right)$
iteration complexity for reducing the duality gap by a constant factor.
Is there a matching upper bound? A major challenge in generalizing
our analysis is understanding the non-local evolution of $\congsimple$.
In our case, the potential function exhibits an ``almost-monotone''
property, with each step changing it by approximately $\delta\sum\congsimple^{3}$.
This enables us to bound the number of short steps, where the $\ell_{3}$
norm of $\congsimple$ is large, by charging them against significant
monotone changes in the potential function. However, for other settings,
such as linear programs, the situation is more complex: $\congsimple$
is obtained by applying a projection matrix to the all-ones vector,
making it difficult to control both its sign and the evolution of
our specific potential function. Overcoming this technical difficulty
likely requires a refined analysis of $\congsimple$'s dynamics, which
could lead to improved iteration bounds for more general IPMs.
\end{itemize}

\section*{Acknowledgements}

This work was partially supported by the French Agence Nationale de
la Recherche (ANR) under grant ANR-21-CE48-0016 (project COMCOPT).
Part of this research was conducted while the author was visiting
the Simons Institute for the Theory of Computing. We are grateful
to Shunhua Jiang and Omri Weinstein for sustained discussions on interior
point methods, as well as to the anonymous reviewers who provided
multiple helpful suggestions.

\appendix

\newpage

\section{Deferred Proofs}
\begin{lem}
\label{lem:general-energy-lemma}Let $A$ be a symmetric M-matrix,
let $r\geq0$ be a vector and let $R=\diag r$. Letting $\rho=\left(A+I\right)^{-1}1$,
one can lower bound
\[
1^{\top}\left(\left(I+R\right)A\left(I+R\right)+I\right)^{-1}1\geq1^{\top}\left(A+I\right)^{-1}1-2\left\langle r,\rho\right\rangle +2\left(\frac{1}{1+\left\Vert r\right\Vert _{\infty}}\right)^{4}\left\langle r,\rho^{2}\right\rangle \,.
\]
\end{lem}

\begin{proof}
We start by lower bounding $\left(\left(I+R\right)A\left(I+R\right)+I\right)^{-1}$
after applying the Woodbury matrix formula. We write
\begin{align*}
 & \left(\left(I+R\right)A\left(I+R\right)+I\right)^{-1}\\
 & =\left(I+R\right)^{-1}\left(A+\left(I+R\right)^{-2}\right)^{-1}\left(I+R\right)^{-1}\\
 & =\left(I+R\right)^{-1}\left(A+I+\diag{\frac{1}{\left(1+r\right)^{2}}-1}\right)^{-1}\left(I+R\right)^{-1}\\
 & =\left(I+R\right)^{-1}\left(\left(A+I\right)^{-1}-\left(A+I\right)^{-1}\left(\left(A+I\right)^{-1}+\diag{\frac{1}{\left(1+r\right)^{2}}-1}^{-1}\right)^{-1}\left(A+I\right)^{-1}\right)\left(I+R\right)^{-1}\\
 & \succeq\left(I+R\right)^{-1}\left(\left(A+I\right)^{-1}-\left(A+I\right)^{-1}\diag{\frac{1}{\left(1+r\right)^{2}}-1}\left(A+I\right)^{-1}\right)\left(I+R\right)^{-1}\\
 & =\left(I+R\right)^{-1}\left(\left(A+I\right)^{-1}+\left(A+I\right)^{-1}\diag{\frac{1}{\left(1+r\right)^{2}}-1}\left(A+I\right)^{-1}\right)\left(I+R\right)^{-1}\\
 & =\left(I+\left(I+R\right)^{-1}-I\right)\left(A+I\right)^{-1}\left(I+\left(I+R\right)^{-1}-I\right)\\
 & +\left(I+R\right)^{-1}\left(A+I\right)^{-1}\diag{1-\frac{1}{\left(1+r\right)^{2}}}\left(A+I\right)^{-1}\left(I+R\right)^{-1}\\
 & =\left(A+I\right)^{-1}-\diag{\frac{r}{1+r}}\left(A+I\right)^{-1}-\left(A+I\right)^{-1}\diag{\frac{r}{1+r}}\\
 & +\diag{\frac{r}{1+r}}\left(A+I\right)^{-1}\diag{\frac{r}{1+r}}\\
 & +\left(I+R\right)^{-1}\left(A+I\right)^{-1}\diag{\frac{2r+r^{2}}{\left(1+r\right)^{2}}}\left(A+I\right)^{-1}\left(I+R\right)^{-1}\\
 & \succeq\left(A+I\right)^{-1}-\diag{\frac{r}{1+r}}\left(A+I\right)^{-1}-\left(A+I\right)^{-1}\diag{\frac{r}{1+r}}\\
 & +\left(I+R\right)^{-1}\left(A+I\right)^{-1}\diag{\frac{2r+r^{2}}{\left(1+r\right)^{2}}}\left(A+I\right)^{-1}\left(I+R\right)^{-1}\,.
\end{align*}
We obtained the last inequality simply by dropping the positive definite
term 
\begin{align*}
\diag{\frac{r}{1+r}}\left(A+I\right)^{-1}\diag{\frac{r}{1+r}}\,.
\end{align*}
Therefore when hitting with the all-ones vector on both sides we obtain:
\begin{align*}
 & 1^{\top}\left(\left(I+R\right)A\left(I+R\right)+I\right)^{-1}1\\
\geq & 1^{\top}\left(A+I\right)^{-1}1-2\left\langle \frac{r}{1+r},\rho\right\rangle +\left(\frac{1}{1+r}\right)^{\top}\left(A+I\right)^{-1}\diag{\frac{2r+r^{2}}{\left(1+r\right)^{2}}}\left(A+I\right)^{-1}\left(\frac{1}{1+r}\right)\,.
\end{align*}
Using the M-matrix property we have that $\left(A+I\right)^{-1}\frac{1}{1+r}\geq\rho\cdot\frac{1}{1+\left\Vert r\right\Vert _{\infty}}$,
pointwise. This enables us to lower bound this quantity further as
\begin{align*}
 & 1^{\top}\left(A+I\right)^{-1}1-2\left\langle \frac{r}{1+r},\rho\right\rangle +\left(\frac{1}{1+\left\Vert r\right\Vert _{\infty}}\right)^{2}\left\Vert \rho\right\Vert _{\diag{\frac{2r+r^{2}}{\left(1+r\right)^{2}}}}^{2}\\
\geq & 1^{\top}\left(A+I\right)^{-1}1-2\left\langle r,\rho\right\rangle +2\left(\frac{1}{1+\left\Vert r\right\Vert _{\infty}}\right)^{4}\left\langle r,\rho^{2}\right\rangle \,,
\end{align*}
which is what we needed.
\end{proof}

\begin{lem}
\label{lem:general-energy-lemma-backward}Let $A$ be a symmetric
M-matrix, let $-1/2\leq r\leq0$ be a vector and let $R=\diag r$.
Letting $\rho=\left(A+I\right)^{-1}1$, one can lower bound
\[
1^{\top}\left(\left(I+R\right)A\left(I+R\right)+I\right)^{-1}1\leq1^{\top}\left(A+I\right)^{-1}1-6\left\langle r,\rho\right\rangle +\frac{3}{2}\left\langle r,\rho^{2}\right\rangle \,.
\]
\end{lem}

\begin{proof}
We start by upper bounding $\left(\left(I+R\right)A\left(I+R\right)+I\right)^{-1}$
after applying the Woodbury matrix formula. We write:
\begin{align*}
 & \left(\left(I+R\right)A\left(I+R\right)+I\right)^{-1}\\
 & =\left(I+R\right)^{-1}\left(A+\left(I+R\right)^{-2}\right)^{-1}\left(I+R\right)^{-1}\\
 & =\left(I+R\right)^{-1}\left(A+I+\diag{\frac{1}{\left(1+r\right)^{2}}-1}\right)^{-1}\left(I+R\right)^{-1}\\
 & =\left(I+R\right)^{-1}\left(\left(A+I\right)^{-1}-\left(A+I\right)^{-1}\left(\left(A+I\right)^{-1}+\diag{\frac{1}{\left(1+r\right)^{2}}-1}^{-1}\right)^{-1}\left(A+I\right)^{-1}\right)\left(I+R\right)^{-1}\\
 & \preceq\left(I+R\right)^{-1}\left(\left(A+I\right)^{-1}-\left(A+I\right)^{-1}\left(I+\diag{\frac{1}{\left(1+r\right)^{2}}-1}^{-1}\right)^{-1}\left(A+I\right)^{-1}\right)\left(I+R\right)^{-1}\\
 & =\left(I+R\right)^{-1}\left(\left(A+I\right)^{-1}-\left(A+I\right)^{-1}\diag{1-\left(1+r\right)^{2}}\left(A+I\right)^{-1}\right)\left(I+R\right)^{-1}\\
 & =\left(I+\left(I+R\right)^{-1}-I\right)\left(A+I\right)^{-1}\left(I+\left(I+R\right)^{-1}-I\right)\\
 & +\left(I+R\right)^{-1}\left(A+I\right)^{-1}\diag{2r+r^{2}}\left(A+I\right)^{-1}\left(I+R\right)^{-1}\\
 & =\left(A+I\right)^{-1}+\left(\left(I+R\right)^{-1}-I\right)\left(A+I\right)^{-1}\\
 & +\left(A+I\right)^{-1}\left(\left(I+R\right)^{-1}-I\right)+\left(\left(I+R\right)^{-1}-I\right)\left(A+I\right)^{-1}\left(\left(I+R\right)^{-1}-I\right)\\
 & +\left(I+R\right)^{-1}\left(A+I\right)^{-1}\diag{2r+r^{2}}\left(A+I\right)^{-1}\left(I+R\right)^{-1}\,.
\end{align*}
Therefore when hitting with $1$ on both sides we obtain:
\begin{align*}
 & 1^{\top}\left(\left(I+R\right)A\left(I+R\right)+I\right)^{-1}1\\
\leq & 1^{\top}\left(A+I\right)^{-1}1-2\left\langle \frac{r}{1+r},\rho\right\rangle +\left(\frac{r}{1+r}\right)^{\top}\left(A+I\right)^{-1}\left(\frac{r}{1+r}\right)\\
+ & \left(\frac{1}{1+r}\right)^{\top}\left(A+I\right)^{-1}\diag{2r+r^{2}}\left(A+I\right)^{-1}\left(\frac{1}{1+r}\right)\,.
\end{align*}
Note that $\left(\frac{r}{1+r}\right)^{\top}\left(A+I\right)^{-1}\left(\frac{r}{1+r}\right)\leq\left\langle \frac{\left|r\right|}{1+r},\rho\right\rangle $
since $\left(A+I\right)^{-1}$ is pointwise non-negative and $0\leq\frac{\left|r\right|}{1+r}\leq1$.
Additionally since $2r+r^{2}\leq\frac{3}{2}r$, we can further bound
this by
\begin{align*}
 & 1^{\top}\left(A+I\right)^{-1}1-3\left\langle \frac{r}{1+r},\rho\right\rangle +\frac{3}{2}\left(\frac{1}{1+r}\right)^{\top}\left(A+I\right)^{-1}\diag r\left(A+I\right)^{-1}\left(\frac{1}{1+r}\right)\\
\leq & 1^{\top}\left(A+I\right)^{-1}1-3\left\langle \frac{r}{1+r},\rho\right\rangle +\frac{3}{2}1^{\top}\left(A+I\right)^{-1}\diag r\left(A+I\right)^{-1}1\\
= & 1^{\top}\left(A+I\right)^{-1}1-3\left\langle \frac{r}{1+r},\rho\right\rangle +\frac{3}{2}\left\langle r,\rho^{2}\right\rangle \\
\leq & 1^{\top}\left(A+I\right)^{-1}1-6\left\langle r,\rho\right\rangle +\frac{3}{2}\left\langle r,\rho^{2}\right\rangle \,.
\end{align*}
\end{proof}
\begin{lem}
[backward local energy increase]\label{lem:local-energy-increase-backward}If
$A$ is a symmetric M-matrix, $x'=x_{\mu}\left(1-\update\right)$,
for some $\update\geq0$, $\mu_{0}\leq\frac{\mu}{1+\delta}\leq2\mu_{0}$
and $1/2\geq\delta>0$ then 

\[
1^{\top}\left(\frac{1}{\mu_{0}}X'AX'+I\right)^{-1}1\leq1^{\top}\left(\frac{1}{\mu_{0}}X_{\mu}AX_{\mu}+I\right)^{-1}1+12\left\Vert \update\right\Vert _{2}\left\Vert \nocong{x_{\mu}}{\frac{\mu}{1+\delta}}\right\Vert _{2}-\frac{3}{4}\left\langle \update,\nocong{x_{\mu}}{\frac{\mu}{1+\delta}}^{2}\right\rangle \,.
\]
\end{lem}

\begin{proof}
First we apply Lemma \ref{lem:general-energy-lemma-backward} on the
matrix to obtain
\begin{align*}
 & 1^{\top}\left(\frac{1}{\mu_{0}}X'AX'+I\right)^{-1}1-1^{\top}\left(\frac{1}{\mu_{0}}X_{\mu}AX_{\mu}+I\right)^{-1}1\\
 & \leq6\left\langle \update,\nocong{x_{\mu}}{\mu_{0}}\right\rangle -\frac{3}{2}\left\langle \update,\nocong{x_{\mu}}{\mu_{0}}^{2}\right\rangle \\
 & \leq6\left\Vert \update\right\Vert _{2}\left\Vert \nocong{x_{\mu}}{\mu_{0}}\right\Vert _{2}-\frac{3}{2}\left\langle \update,\nocong{x_{\mu}}{\mu_{0}}^{2}\right\rangle \,.
\end{align*}
Applying Lemma \ref{lem:rho-mu-stability} we can bound $\left\Vert \nocong{x_{\mu}}{\mu_{0}}\right\Vert _{2}\leq\frac{\mu}{\mu_{0}\left(1+\delta\right)}\left\Vert \nocong x{\frac{\mu}{1+\delta}}\right\Vert _{2}\leq2\left\Vert \nocong x{\frac{\mu}{1+\delta}}\right\Vert _{2}$
and $\nocong{x_{\mu}}{\mu_{0}}\geq\frac{\mu_{0}\left(1+\delta\right)}{\mu}\nocong x{\mu}\geq\frac{1}{2}\nocong x{\mu}$.
Thus we can continue this chain of inequalities with 
\[
\leq12\left\Vert \update\right\Vert _{2}\left\Vert \nocong{x_{\mu}}{\frac{\mu}{1+\delta}}\right\Vert _{2}-\frac{3}{4}\left\langle \update,\nocong{x_{\mu}}{\frac{\mu}{1+\delta}}^{2}\right\rangle \,,
\]
which concludes the proof.
\end{proof}
\begin{lem}
[backward energy lemma]\label{lem:energy-lemma-scaling-backward}Let
$x_{\mu}$ and $x_{\mu'}$ be points on the central path, for $\mu'=\frac{\mu}{1+\delta}$,
$1/2>\delta>0$ such that $\left\Vert \delta\nocong{x_{\mu}}{\frac{\mu}{1+\delta}}\right\Vert _{3}=\alpha\leq\frac{1}{16}$.
Then
\[
\Phi_{\mu'}\leq\Phi_{\mu}+2^{5}\alpha n^{1/6}\left\Vert \nocong{x_{\mu}}{\frac{\mu}{1+\delta}}\right\Vert _{2}-2^{-1}\alpha\left\Vert \nocong{x_{\mu}}{\frac{\mu}{1+\delta}}\right\Vert _{3}^{2}\,.
\]
\end{lem}

\begin{proof}
To establish this lower bound on the decrease in potential, we first
use Lemma \ref{lem:full-correction-update} which shows that the multiplicative
change in $x_{\mu}$ determined by the Newton predictor step $x'=x_{\mu}\left(1-\delta\nocong{x_{\mu}}{\frac{\mu}{1+\delta}}\right)$
approximates very well the multiplicative change encountered when
moving straight to $x_{\mu\left(1-\delta\right)}$. Indeed, for $\delta$
chosen such that one has that $\left\Vert \delta\nocong{x_{\mu}}{\frac{\mu}{1+\delta}}\right\Vert _{3}=\alpha\leq1/16$,
we have
\[
x_{\frac{\mu}{1+\delta}}=x_{\mu}\left(1-\delta\nocong{x_{\mu}}{\frac{\mu}{1+\delta}}-\vzeta\right)
\]
where $\left\Vert \vzeta\right\Vert _{2}\leq8\alpha^{2}$. Now we
apply Lemma \ref{lem:local-energy-increase-backward}, to upper bound
(recalling that $\mu_{0}\leq\frac{\mu}{1+\delta}\leq2\mu_{0}$):
\begin{align*}
 & 1^{\top}\left(\frac{1}{\mu_{0}}X_{\frac{\mu}{1+\delta}}AX_{\frac{\mu}{1+\delta}}+I\right)^{-1}1-1^{\top}\left(\frac{1}{\mu_{0}}X_{\mu}AX_{\mu}+I\right)^{-1}1\\
\leq & 12\left\Vert \delta\nocong{x_{\mu}}{\frac{\mu}{1+\delta}}+\vzeta\right\Vert _{2}\left\Vert \nocong{x_{\mu}}{\frac{\mu}{1+\delta}}\right\Vert _{2}-\frac{3}{4}\left\langle \delta\nocong{x_{\mu}}{\frac{\mu}{1+\delta}}+\vzeta,\nocong{x_{\mu}}{\frac{\mu}{1+\delta}}^{2}\right\rangle \\
\leq & 12\left(\delta\left\Vert \nocong{x_{\mu}}{\frac{\mu}{1+\delta}}\right\Vert _{2}+\left\Vert \vzeta\right\Vert _{2}\right)\left\Vert \nocong{x_{\mu}}{\frac{\mu}{1+\delta}}\right\Vert _{2}-\frac{3}{4}\delta\left\Vert \nocong{x_{\mu}}{\frac{\mu}{1+\delta}}\right\Vert _{3}^{3}-\frac{3}{4}\left\Vert \vzeta\right\Vert _{2}\left\Vert \nocong{x_{\mu}}{\frac{\mu}{1+\delta}}\right\Vert _{4}^{2}\\
\leq & 12\left(\alpha n^{1/6}+8\alpha^{2}\right)\left\Vert \nocong{x_{\mu}}{\frac{\mu}{1+\delta}}\right\Vert _{2}-\frac{3}{4}\left(\alpha-8\alpha^{2}\right)\left\Vert \nocong{x_{\mu}}{\frac{\mu}{1+\delta}}\right\Vert _{3}^{2}\\
\leq & 2^{5}\alpha n^{1/6}\left\Vert \nocong{x_{\mu}}{\frac{\mu}{1+\delta}}\right\Vert _{2}-\frac{1}{2}\alpha\left\Vert \nocong{x_{\mu}}{\frac{\mu}{1+\delta}}\right\Vert _{3}^{2}.
\end{align*}
\end{proof}
\begin{lem}
[full correction update]\label{lem:full-correction-update} Let $x_{\mu}$
be a central point, and let $x'=x_{\mu}\left(1+\cong{x_{\mu}}{\frac{\mu}{1-\delta}}\right)$
be the point reached after executing a Newton step, for $\delta>0$
chosen such that $\left\Vert \cong{x_{\mu}}{\frac{\mu}{1-\delta}}\right\Vert _{3}\leq\alpha\leq1/2$.
Then 
\[
x_{\frac{\mu}{1-\delta}}=x\left(1+\cong{x_{\mu}}{\frac{\mu}{1-\delta}}+\vzeta\right)
\]
where $\left\Vert \vzeta\right\Vert _{2}\leq8\alpha^{2}$.
\end{lem}

\begin{proof}
We use the fact that $x_{\frac{\mu}{1-\delta}}$ can be found by executing
subsequent Newton correction steps after computing $x'$. Since these
subsequent correction steps have very small norm per Lemma \ref{lem:correction-lemma},
the multiplicative update $1+\delta\nocong{x_{\mu}}{\mu'}$ computed
during the predictor step approximates very well the multiplicative
update that takes us from $x_{\mu}$ to $x_{\mu'}$. Indeed, Lemma
\ref{lem:correction-lemma} shows that after moving to 
\[
x'=x_{\mu}\left(1+\cong{x_{\mu}}{\frac{\mu}{1-\delta}}\right)\,,
\]
the new $\cong{x'}{\frac{\mu}{1-\delta}}$ satisfies $\left\Vert \cong{x'}{\frac{\mu}{1-\delta}}\right\Vert _{2}\leq\left\Vert \cong{x_{\mu}}{\frac{\mu}{1-\delta}}\right\Vert _{4}^{2}\,.$Using
this fact we see that
\[
x_{\mu'}=x_{\mu}\prod_{t\geq0}\left(1+\cong{x^{\left(t\right)}}{\frac{\mu}{1-\delta}}\right)\,.
\]
where $x^{\left(0\right)}=x_{\mu}$ and $\left\Vert \cong{x^{\left(t\right)}}{\frac{\mu}{1-\delta}}\right\Vert _{2}\leq\left\Vert \cong{x^{\left(t-1\right)}}{\frac{\mu}{1-\delta}}\right\Vert _{2}^{2}$.
Thus after expanding we have
\begin{align*}
x_{\mu'} & =x_{\mu}\left(1+\cong{x_{\mu}}{\frac{\mu}{1-\delta}}\right)\prod_{t\geq1}\left(1+\cong{x^{\left(t\right)}}{\frac{\mu}{1-\delta}}\right)\\
 & =x_{\mu}\left(1+\cong{x_{\mu}}{\frac{\mu}{1-\delta}}+\underbrace{\left(1+\cong{x_{\mu}}{\frac{\mu}{1-\delta}}\right)\left(\prod_{t\geq1}\left(1+\cong{x^{\left(t\right)}}{\frac{\mu}{1-\delta}}\right)-1\right)}_{\vzeta}\right)\,.
\end{align*}
We also have that 
\[
1-\left|\sum_{t\geq1}\cong{x^{\left(t\right)}}{\frac{\mu}{1-\delta}}\right|\leq\prod_{t\geq1}\left(1+\cong{x^{\left(t\right)}}{\frac{\mu}{1-\delta}}\right)\leq\exp\left(\sum_{t\geq1}\cong{x^{\left(t\right)}}{\frac{\mu}{1-\delta}}\right)\leq1+2\left|\sum_{t\geq1}\cong{x^{\left(t\right)}}{\frac{\mu}{1-\delta}}\right|
\]
where we use the fact that we can bound $\left\Vert \sum_{t\geq1}\cong{x^{\left(t\right)}}{\frac{\mu}{1-\delta}}\right\Vert _{\infty}\leq\sum_{t\geq1}\left\Vert \cong{x^{\left(t\right)}}{\frac{\mu}{1-\delta}}\right\Vert _{2}\leq\alpha^{2}+\alpha^{4}+\alpha^{8}+\dots\leq2\alpha^{2}$.
Finally we can bound the norm of $\zeta$ via 
\begin{align*}
\left\Vert \zeta\right\Vert _{2} & \leq\left\Vert 1+\cong{x_{\mu}}{\frac{\mu}{1-\delta}}\right\Vert _{2}\left\Vert \prod_{t\geq1}\left(1+\cong{x^{\left(t\right)}}{\frac{\mu}{1-\delta}}\right)-1\right\Vert _{2}\\
 & \leq\left\Vert 1+\cong{x_{\mu}}{\frac{\mu}{1-\delta}}\right\Vert _{2}\left\Vert 2\sum_{t\geq1}\left|\cong{x^{\left(t\right)}}{\frac{\mu}{1-\delta}}\right|\right\Vert _{2}\\
 & \leq2\left\Vert 1+\cong{x_{\mu}}{\frac{\mu}{1-\delta}}\right\Vert _{2}\sum_{t\geq1}\left\Vert \cong{x^{\left(t\right)}}{\frac{\mu}{1-\delta}}\right\Vert _{2}\\
 & \leq2\cdot\left(1+\alpha\right)\cdot2\alpha^{2}\\
 & \leq8\alpha^{2}\,.
\end{align*}
\end{proof}

\section{Using Fast Linear System Solvers for M-Matrices\label{sec:mmatrix-solver}}

In this section we prove Corollary \ref{cor:runtime}. First we require
a lemma that shows that if $x$ approximately scales a matrix $A$,
then the norm of $x$ can not be too large.
\begin{lem}
\label{lem:x-bound}Let $\varepsilon\in\left(0,1\right)$, and let
$x$ be a vector that approximately scales a positive semidefinite
matrix $A$, in the sense that $\left\Vert XAX1-1\right\Vert _{2}\leq\varepsilon$.
Then
\[
\left\Vert x\right\Vert _{2}^{2}\leq\frac{2n}{\lambda_{\min}\left(A\right)}\,.
\]
\end{lem}

\begin{proof}
By the hypothesis we have that 
\[
XAX1=1+\zeta\,,
\]
where $\left\Vert \zeta\right\Vert _{2}\leq\varepsilon$. Hitting
both sides on the left with $1^{\top}$ we obtain
\[
x^{\top}Ax=n+\left\langle 1,\zeta\right\rangle \,.
\]
Since $x^{\top}Ax\geq\left\Vert x\right\Vert _{2}^{2}\cdot\lambda_{\min}\left(A\right)$
we conclude that 
\[
\left\Vert x\right\Vert _{2}^{2}\leq\frac{n+\left\langle 1,\zeta\right\rangle }{\lambda_{\min}\left(A\right)}\leq\frac{n+\varepsilon\sqrt{n}}{\lambda_{\min}\left(A\right)}\leq\frac{2n}{\lambda_{\min}\left(A\right)}\,.
\]
\end{proof}
The following lemma shows that by solving matrix scaling on an input
perturbed with a small multiple of the identity, we achieve a solution
of comparable quality for the original problem.
\begin{lem}
\label{lem:scaling-perturbed-input}Let $A$ be a symmetric M-matrix.
Let $A'=\gamma I+A$, where $\gamma=\frac{\varepsilon}{4n}\lambda_{\min}\left(A\right)$.
If $x$ satisfies $\left\Vert XA'X1-1\right\Vert _{2}\leq\varepsilon/2$,
then $\left\Vert XAX1-1\right\Vert _{2}\leq\varepsilon$.
\end{lem}

\begin{proof}
By construction, $\lambda_{\min}\left(A'\right)\geq\lambda_{\min}\left(A\right)+\gamma$.
From Lemma \ref{lem:x-bound} we know that 
\[
\left\Vert x\right\Vert _{2}^{2}\leq\frac{2n}{\lambda_{\min}\left(A\right)+\gamma}\,,
\]
which allows us to bound
\begin{align*}
\left\Vert XAX1-1\right\Vert _{2} & =\left\Vert X\left(A'-\gamma I\right)X1-1\right\Vert _{2}=\left\Vert \left(XA'X1-1\right)-\gamma X^{2}1\right\Vert _{2}\\
 & \leq\left\Vert XA'X1-1\right\Vert _{2}+\gamma\left\Vert x\right\Vert _{4}^{2}\leq\varepsilon/2+\gamma\frac{2n}{\lambda_{\min}\left(A\right)+\gamma}\\
 & =\varepsilon/2+\frac{2n}{\frac{\lambda_{\min}\left(A\right)}{\gamma}+1}\leq\varepsilon/2+\frac{2n}{\frac{\lambda_{\min}\left(A\right)}{\gamma}}\\
 & =\varepsilon\,.
\end{align*}
\end{proof}
Finally, this allows us to obtain a diagonal rescaling of the input
matrix which is SDD, and for which Laplacian system solvers can provide
high-precision solutions. We use the following theorem from \cite{ahmadinejad2019perron}.
\begin{thm}
[Theorem 8 from \cite{ahmadinejad2019perron}, arXiv version]\label{thm:perron-solver}
Let $A\in\mathbb{R}^{n\times n}$ be a symmetric nonnegative matrix
with $m$ nonzero entries and let $s>0$ be such that $\rho\left(A\right)<s$.
Let $M=sI-A$. Pick $\varepsilon>0$. Then SymMMatrix-Scale$\left(A,\varepsilon,K\right)$
computes a diagonal matrix $V$ where $V\left(\left(1+\varepsilon\right)sI-A\right)V$
is SDD with high probability. This algorithm runs in time
\[
O\left(\mathcal{T}_{\text{SDD solve}}\left(m,n,1/4\right)\log\left(n\kappa\left(M\right)\right)\log\left(1/\varepsilon\right)\right)\,.
\]
\end{thm}

We can now prove the first part of Corollary \ref{cor:runtime}.
\begin{proof}
[Proof of Corollary \ref{cor:runtime}, first part] We first approximate
the largest eigenvalue of $A$ using the power method, to obtain a
decomposition $A=sI-C$, with $s\leq2\lambda_{\max}\left(A\right)$,
and $\rho\left(C\right)<s$. Define $A'=A+\varepsilon'sI$, where
$\varepsilon'=\frac{\varepsilon}{8n\kappa\left(A\right)}$, and run
the algorithm from Theorem \ref{thm:perron-solver} on input $C$
with error parameter $\varepsilon'$. We obtain in time
\[
O\left(\mathcal{T}_{\text{SDD solve}}\left(\text{nnz}\left(A\right),n,1/4\right)\log\left(n\kappa\left(A\right)\right)\log\left(n\kappa\left(A\right)/\varepsilon\right)\right)
\]
a diagonal matrix $V$ where 
\[
V\left(\left(1+\varepsilon'\right)sI-C\right)V=VA'V
\]
is SDD. Since the magnitude of the perturbation $\varepsilon'sI=A'-A$
is at most $\frac{\varepsilon}{8n\kappa\left(A\right)}\cdot2\lambda_{\max}\left(A\right)=\frac{\varepsilon}{4n}\cdot\lambda_{\min}\left(A\right)$,
by Lemma \ref{lem:scaling-perturbed-input} finding a $\varepsilon/2$
scaling for $A'$ yields an $\varepsilon$ scaling for the original
$A$. By Theorem \ref{thm:main-scaling}, in $O\left(n^{1/3}\log\frac{\left\Vert \left(A+\varepsilon'sI\right)1-1\right\Vert _{2}}{\varepsilon}\right)$
IPM iterations we obtain the target scaling $x$. 

The key to solving each IPM iteration efficiently is to reuse the
diagonal scaling $V$ that we compute only once. Each linear system
solve of the algorithm requires solving linear systems involving a
matrix of the form 
\begin{align*}
\frac{1}{\mu}XA'X+I & =\frac{1}{\mu}XV^{-1}\left(VA'V\right)V^{-1}X+I\\
 & =XV^{-1}\left(\frac{1}{\mu}VA'V+VX^{-2}V\right)V^{-1}X\,.
\end{align*}
Since $VA'V$ is SDD, the the matrix in the middle is also SDD, so
standard Laplacian system solvers can be used. Hence we obtain a total
runtime of 
\begin{align*}
 & O\left(\mathcal{T}_{\text{SDD solve}}\left(\text{nnz}\left(A\right),n,1/4\right)\log\left(n\kappa\left(A\right)\right)\log\left(n\kappa\left(A\right)/\varepsilon\right)\right)\\
+ & O\left(\mathcal{T}_{\text{SDD solve}}\left(\text{nnz}\left(A\right),n,1/100\right)\cdot n^{1/3}\log\left(\frac{\left\Vert A1-1\right\Vert _{2}}{\varepsilon}+n\lambda_{\min}\left(A\right)\right)\right)
\end{align*}
which we can bound as 
\[
\widetilde{O}\left(\text{nnz}\left(A\right)\cdot\left(\log\kappa\left(A\right)\cdot\log\frac{\kappa\left(A\right)}{\varepsilon}+n^{1/3}\log\frac{\left\Vert A1-1\right\Vert _{2}}{\varepsilon}\right)\right)
\]
after suppressing $\text{polylog}\left(n\right)$ factors.
\end{proof}
For the second part we similarly show that slightly perturbing the
input matrix does not hurt the error guarantees by too much. First
we require a lemma that permits us to upper bound the $\ell_{2}$
norm of any point on the central path.
\begin{lem}
\label{lem:diff-perturbed-input}Let $A\in\mathbb{R}^{n\times n}$
be a symmetric M-matrix. If $x$ satisfies the $\mu$-centrality condition
\[
\frac{1}{\mu}\left(Ax-b\right)-\frac{1}{x}=0\,,
\]
then
\[
\left\Vert x\right\Vert _{2}^{2}\leq\max\left\{ \frac{2\varepsilon}{\lambda_{\min}\left(A\right)},\left(\frac{2\left\Vert b\right\Vert _{2}}{\lambda_{\min}\left(A\right)}\right)^{2}\right\} \,.
\]
\end{lem}

\begin{proof}
By hitting the equation with $x^{\top}$ we obtain that 
\begin{align*}
x^{\top}Ax-b^{\top}x & =\mu n\,.
\end{align*}
Therefore 
\[
x^{\top}Ax=\mu n+b^{\top}x\geq\left\Vert x\right\Vert _{2}^{2}\lambda_{\min}\left(A\right)\,.
\]
Thus
\[
\left\Vert x\right\Vert _{2}^{2}\leq\frac{\mu n+b^{\top}x}{\lambda_{\min}\left(A\right)}\leq2\max\left\{ \frac{\mu n}{\lambda_{\min}\left(A\right)},\frac{\left\Vert b\right\Vert _{2}\left\Vert x\right\Vert _{2}}{\lambda_{\min}\left(A\right)}\right\} \,,
\]
and so
\[
\left\Vert x\right\Vert _{2}^{2}\leq\max\left\{ \frac{2\mu n}{\lambda_{\min}\left(A\right)},\left(\frac{2\left\Vert b\right\Vert _{2}}{\lambda_{\min}\left(A\right)}\right)^{2}\right\} \,.
\]
\end{proof}
Now we can quantitatively control how much a perturbation of the input
matrix by a multiple of the identity changes error guarantees.
\begin{lem}
\label{lem:diff-perturbed-input-2}Let $A$ be a symmetric M-matrix.
Let $A'=\gamma I+A$, for $\gamma\leq\frac{\varepsilon}{8}\left(\frac{\lambda_{\min}\left(A\right)}{\left\Vert b\right\Vert _{2}}\right)^{2}$.
Denote by $Q_{A}\left(x\right)=\frac{1}{2}x^{\top}Ax-b^{\top}x$,
and let $x_{A}^{*}$ and $x_{A'}^{*}$ be the minimizers of $Q_{A}$
and $Q_{A'}$, respectively, over the non-negative quadrant. Let $x\geq0$
such that 
\[
Q_{A'}\left(x\right)\leq Q_{A'}\left(x_{A'}^{*}\right)+\varepsilon/2\,.
\]
Then
\[
Q_{A}\left(x\right)\leq Q_{A}\left(x_{A}^{*}\right)+\varepsilon\,.
\]
\end{lem}

\begin{proof}
From the definition, we immediately have that: 
\[
Q_{A'}\left(x_{A'}^{*}\right)\leq Q_{A'}\left(x_{A}^{*}\right)=Q_{A}\left(x_{A}^{*}\right)+\frac{\gamma}{2}\left\Vert x_{A}^{*}\right\Vert _{2}^{2}\,.
\]
Applying Lemma \ref{lem:diff-perturbed-input} with $\mu\rightarrow0$
we obtain that 
\[
\left\Vert x_{A}^{*}\right\Vert _{2}^{2}\leq\left(\frac{2\left\Vert b\right\Vert _{2}}{\lambda_{\min}\left(A\right)}\right)^{2}\,,
\]
which allows us to bound the amount by which the objective value at
the minimizer increases when perturbing $A$. Therefore 
\begin{align*}
Q_{A}\left(x\right)-Q_{A}\left(x_{A}^{*}\right) & \leq Q_{A'}\left(x\right)-Q_{A}\left(x_{A}^{*}\right)\\
 & =\left(Q_{A'}\left(x\right)-Q_{A'}\left(x_{A}^{*}\right)\right)+\left(Q_{A'}\left(x_{A}^{*}\right)-Q_{A}\left(x_{A}^{*}\right)\right)\\
 & \leq\left(Q_{A'}\left(x\right)-Q_{A'}\left(x_{A'}^{*}\right)\right)+\left(Q_{A'}\left(x_{A}^{*}\right)-Q_{A}\left(x_{A}^{*}\right)\right)\\
 & \leq\varepsilon/2+\frac{\gamma}{2}\left\Vert x_{A}^{*}\right\Vert _{2}^{2}\\
 & \leq\varepsilon\,.
\end{align*}
\end{proof}
Finally, we can prove the second part of Corollary \ref{cor:runtime}.
\begin{proof}
[Proof of Corollary \ref{cor:runtime}, second part] Similarly to
the first part, we approximate the largest eigenvalue of $A$ using
the power method, to obtain a decomposition $A=sI-C$, with $s\leq2\lambda_{\max}\left(A\right)$,
and $\rho\left(C\right)<s$. Define $A'=A+\varepsilon'sI$, where
$\varepsilon'=\frac{\varepsilon\lambda_{\min}\left(A\right)}{16\kappa\left(A\right)\left\Vert b\right\Vert _{2}^{2}}$,
and run the algorithm from Theorem \ref{thm:perron-solver} on input
$C$ with error parameter $\varepsilon'$. Again, we obtain in time
\[
O\left(\mathcal{T}_{\text{SDD solve}}\left(\text{nnz}\left(A\right),n,1/4\right)\log\left(n\kappa\left(A\right)\right)\log\left(1/\varepsilon'\right)\right)
\]
a diagonal matrix $V$ where 
\[
V\left(\left(1+\varepsilon'\right)sI-C\right)V=VA'V
\]
is SDD. Since the magnitude of the perturbation $\varepsilon'sI=A'-A$
is at most $\frac{\varepsilon\lambda_{\min}\left(A\right)}{16\kappa\left(A\right)\left\Vert b\right\Vert _{2}^{2}}\cdot2\lambda_{\max}\left(A\right)=\frac{\varepsilon}{8}\left(\frac{\lambda_{\min}\left(A\right)}{\left\Vert b\right\Vert _{2}}\right)^{2}$,
by Lemma \ref{lem:diff-perturbed-input-2} solving for $A'$ to $\varepsilon/2$
additive error yields an $\varepsilon$ additive error for the original
$A$. By Theorem \ref{thm:main-quadratic}, in $O\left(n^{1/3}\log\frac{\left\Vert \left(A+\varepsilon'sI\right)1-1\right\Vert _{2}+\left\Vert b\right\Vert _{2}}{\varepsilon}\right)$
IPM iterations we obtain the target iterate $x$. 

Just as before, we can solve each IPM iteration efficiently by reusing
the diagonal scaling $V$ which we compute only once. Hence we obtain
a total runtime of 
\begin{align*}
 & O\left(\mathcal{T}_{\text{SDD solve}}\left(\text{nnz}\left(A\right),n,1/4\right)\log\left(n\kappa\left(A\right)\right)\log\left(\frac{\kappa\left(A\right)\left\Vert b\right\Vert _{2}}{\lambda_{\min}\left(A\right)\varepsilon}\right)\right)\\
+ & O\left(\mathcal{T}_{\text{SDD solve}}\left(\text{nnz}\left(A\right),n,1/100\right)\cdot n^{1/3}\log\left(\frac{\left\Vert A1-1\right\Vert _{2}+\left\Vert b\right\Vert _{2}}{\varepsilon}+\frac{\lambda_{\min}\left(A\right)}{\left\Vert b\right\Vert _{2}}\right)\right)
\end{align*}
which we can bound as 
\[
\widetilde{O}\left(\text{nnz}\left(A\right)\cdot\left(\log\kappa\left(A\right)\cdot\log\frac{\kappa\left(A\right)\left\Vert b\right\Vert _{2}}{\varepsilon\lambda_{\min}\left(A\right)}+n^{1/3}\log\left(\frac{\left\Vert A1-1\right\Vert _{2}+\left\Vert b\right\Vert _{2}}{\varepsilon}+\frac{\lambda_{\min}\left(A\right)}{\left\Vert b\right\Vert _{2}}\right)\right)\right)
\]
after suppressing $\text{polylog}\left(n\right)$ factors.
\end{proof}

\bibliographystyle{alpha}
\bibliography{ref}

\newcommand{\etalchar}[1]{$^{#1}$}
\begin{thebibliography}{vdBLN{\etalchar{+}}20}

\bibitem[AJSS19]{ahmadinejad2019perron}
AmirMahdi Ahmadinejad, Arun Jambulapati, Amin Saberi, and Aaron Sidford.
\newblock Perron-frobenius theory in nearly linear time: Positive eigenvectors,
  m-matrices, graph kernels, and other applications.
\newblock In {\em Proceedings of the Thirtieth Annual ACM-SIAM Symposium on
  Discrete Algorithms}, pages 1387--1404. SIAM, 2019.

\bibitem[AMV20]{axiotis2020circulation}
Kyriakos Axiotis, Aleksander M{\k{a}}dry, and Adrian Vladu.
\newblock Circulation control for faster minimum cost flow in unit-capacity
  graphs.
\newblock In {\em 61st Annual Symposium on Foundations of Computer Science},
  pages 93--104. IEEE, 2020.

\bibitem[AMV22]{axiotis2022faster}
Kyriakos Axiotis, Aleksander M{\k{a}}dry, and Adrian Vladu.
\newblock Faster sparse minimum cost flow by electrical flow localization.
\newblock In {\em 2021 IEEE 62nd Annual Symposium on Foundations of Computer
  Science (FOCS)}, pages 528--539. IEEE, 2022.

\bibitem[AZLOW17]{allen2017much}
Zeyuan Allen-Zhu, Yuanzhi Li, Rafael Oliveira, and Avi Wigderson.
\newblock Much faster algorithms for matrix scaling.
\newblock In {\em 2017 IEEE 58th Annual Symposium on Foundations of Computer
  Science (FOCS)}, pages 890--901. IEEE, 2017.

\bibitem[BCLL18]{bubeck2018homotopy}
S{\'e}bastien Bubeck, Michael~B Cohen, Yin~Tat Lee, and Yuanzhi Li.
\newblock An homotopy method for lp regression provably beyond self-concordance
  and in input-sparsity time.
\newblock In {\em Proceedings of the 50th Annual ACM SIGACT Symposium on Theory
  of Computing}, pages 1130--1137, 2018.

\bibitem[BE15]{bubeck2014entropic}
S{\'e}bastien Bubeck and Ronen Eldan.
\newblock The entropic barrier: a simple and optimal universal self-concordant
  barrier.
\newblock In {\em Conference on Learning Theory}, pages 279--279. PMLR, 2015.

\bibitem[CKL{\etalchar{+}}22]{chen2022maximum}
Li~Chen, Rasmus Kyng, Yang~P Liu, Richard Peng, Maximilian~Probst Gutenberg,
  and Sushant Sachdeva.
\newblock Maximum flow and minimum-cost flow in almost-linear time.
\newblock In {\em 2022 IEEE 63rd Annual Symposium on Foundations of Computer
  Science (FOCS)}, pages 612--623. IEEE, 2022.

\bibitem[CLM{\etalchar{+}}16]{cohen2016geometric}
Michael~B Cohen, Yin~Tat Lee, Gary Miller, Jakub Pachocki, and Aaron Sidford.
\newblock Geometric median in nearly linear time.
\newblock In {\em Proceedings of the forty-eighth annual ACM symposium on
  Theory of Computing}, pages 9--21, 2016.

\bibitem[CLS21]{cohen2021solving}
Michael~B Cohen, Yin~Tat Lee, and Zhao Song.
\newblock Solving linear programs in the current matrix multiplication time.
\newblock {\em Journal of the ACM (JACM)}, 68(1):1--39, 2021.

\bibitem[CMSV17]{cmsv17}
Michael~B Cohen, Aleksander M\k{a}dry, Piotr Sankowski, and Adrian Vladu.
\newblock Negative-weight shortest paths and unit capacity minimum cost flow in
  $\widetilde{O}(m^{10/7} \log {W})$ time.
\newblock In {\em Proceedings of the 28th Annual ACM-SIAM Symposium on Discrete
  Algorithms}, pages 752--771. SIAM, 2017.

\bibitem[CMTV17]{cohen2017matrix}
Michael~B Cohen, Aleksander Madry, Dimitris Tsipras, and Adrian Vladu.
\newblock Matrix scaling and balancing via box constrained newton's method and
  interior point methods.
\newblock In {\em 2017 IEEE 58th Annual Symposium on Foundations of Computer
  Science (FOCS)}, pages 902--913. IEEE, 2017.

\bibitem[CPW22]{chen20222}
Li~Chen, Richard Peng, and Di~Wang.
\newblock 2-norm flow diffusion in near-linear time.
\newblock In {\em 2021 IEEE 62nd Annual Symposium on Foundations of Computer
  Science (FOCS)}, pages 540--549. IEEE, 2022.

\bibitem[DS08]{daitch2008faster}
Samuel~I Daitch and Daniel~A Spielman.
\newblock Faster approximate lossy generalized flow via interior point
  algorithms.
\newblock In {\em Proceedings of the fortieth annual ACM symposium on Theory of
  computing}, pages 451--460, 2008.

\bibitem[FWY20]{fountoulakis2020p}
Kimon Fountoulakis, Di~Wang, and Shenghao Yang.
\newblock P-norm flow diffusion for local graph clustering.
\newblock In {\em International Conference on Machine Learning}, pages
  3222--3232. PMLR, 2020.

\bibitem[GLP23]{gao2021fully}
Yu~Gao, Yang Liu, and Richard Peng.
\newblock Fully dynamic electrical flows: Sparse maxflow faster than
  {G}oldberg--{R}ao.
\newblock {\em SIAM Journal on Computing}, 0(0):85--156, 2023.

\bibitem[Gon12]{gondzio2012interior}
Jacek Gondzio.
\newblock Interior point methods 25 years later.
\newblock {\em European Journal of Operational Research}, 218(3):587--601,
  2012.

\bibitem[Hil14]{hildebrand2014canonical}
Roland Hildebrand.
\newblock Canonical barriers on convex cones.
\newblock {\em Mathematics of operations research}, 39(3):841--850, 2014.

\bibitem[JK06]{jin2006procedure}
Yi~Jin and Bahman Kalantari.
\newblock A procedure of chv{\'a}tal for testing feasibility in linear
  programming and matrix scaling.
\newblock {\em Linear algebra and its applications}, 416(2-3):795--798, 2006.

\bibitem[Joh82]{johnson1982inverse}
Charles~R Johnson.
\newblock Inverse m-matrices.
\newblock {\em Linear Algebra and its Applications}, 47:195--216, 1982.

\bibitem[KK92]{khachiyan1992diagonal}
Leonid Khachiyan and Bahman Kalantari.
\newblock Diagonal matrix scaling and linear programming.
\newblock {\em SIAM Journal on Optimization}, 2(4):668--672, 1992.

\bibitem[KLP{\etalchar{+}}16]{kyng2016sparsified}
Rasmus Kyng, Yin~Tat Lee, Richard Peng, Sushant Sachdeva, and Daniel~A
  Spielman.
\newblock Sparsified cholesky and multigrid solvers for connection laplacians.
\newblock In {\em Proceedings of the forty-eighth annual ACM symposium on
  Theory of Computing}, pages 842--850, 2016.

\bibitem[KLS20]{liu2020faster}
Tarun Kathuria, Yang~P. Liu, and Aaron Sidford.
\newblock Unit capacity maxflow in almost ${O}(m^{4/3})$ time.
\newblock In {\em 61st Annual Symposium on Foundations of Computer Science},
  pages 119--130. IEEE, 2020.

\bibitem[LS14]{lee2014path}
Yin~Tat Lee and Aaron Sidford.
\newblock Path finding methods for linear programming: Solving linear programs
  in $\widetilde{O}(\sqrt{rank})$ iterations and faster algorithms for maximum
  flow.
\newblock In {\em 2014 IEEE 55th Annual Symposium on Foundations of Computer
  Science}, pages 424--433. IEEE, 2014.

\bibitem[LS20]{liu2019faster}
Yang~P Liu and Aaron Sidford.
\newblock Faster energy maximization for faster maximum flow.
\newblock In {\em Proceedings of the 52nd Annual ACM SIGACT Symposium on Theory
  of Computing}, 2020.

\bibitem[M{\k{a}}d13]{madry2013navigating}
Aleksander M{\k{a}}dry.
\newblock Navigating central path with electrical flows: From flows to
  matchings, and back.
\newblock In {\em 54th Annual Symposium on Foundations of Computer Science},
  pages 253--262. IEEE, 2013.

\bibitem[M{\k{a}}d16]{madry2016computing}
Aleksander M{\k{a}}dry.
\newblock Computing maximum flow with augmenting electrical flows.
\newblock In {\em 57th Annual Symposium on Foundations of Computer Science},
  pages 593--602. IEEE, 2016.

\bibitem[NN94]{nesterov1994interior}
Yurii Nesterov and Arkadii Nemirovskii.
\newblock {\em Interior-point polynomial algorithms in convex programming}.
\newblock SIAM, 1994.

\bibitem[PS14]{peng2014efficient}
Richard Peng and Daniel~A Spielman.
\newblock An efficient parallel solver for sdd linear systems.
\newblock In {\em Proceedings of the 46th Annual ACM Symposium on Theory of
  computing}, pages 333--342, 2014.

\bibitem[SZ23]{sachdeva2023simple}
Sushant Sachdeva and Yibin Zhao.
\newblock A simple and efficient parallel laplacian solver.
\newblock In {\em Proceedings of the 35th ACM Symposium on Parallelism in
  Algorithms and Architectures}, pages 315--325, 2023.

\bibitem[Tod20]{todd2020lower}
Michael~J Todd.
\newblock A lower bound on the number of iterations of an interior-point
  algorithm for linear programming.
\newblock In {\em Numerical Analysis 1993}, pages 237--259. CRC Press, 2020.

\bibitem[TY96]{todd1996lower}
Michael~J Todd and Yinyu Ye.
\newblock A lower bound on the number of iterations of long-step primal-dual
  linear programming algorithms.
\newblock {\em Annals of Operations Research}, 62(1):233--252, 1996.

\bibitem[vdBGJ{\etalchar{+}}22]{van2022faster}
Jan van~den Brand, Yu~Gao, Arun Jambulapati, Yin~Tat Lee, Yang~P Liu, Richard
  Peng, and Aaron Sidford.
\newblock Faster maxflow via improved dynamic spectral vertex sparsifiers.
\newblock In {\em Proceedings of the 54th Annual ACM SIGACT Symposium on Theory
  of Computing}, pages 543--556, 2022.

\bibitem[vdBLL{\etalchar{+}}21]{van2021minimum}
Jan van~den Brand, Yin~Tat Lee, Yang~P Liu, Thatchaphol Saranurak, Aaron
  Sidford, Zhao Song, and Di~Wang.
\newblock Minimum cost flows, mdps, and $\ell_1$-regression in nearly linear
  time for dense instances.
\newblock In {\em Proceedings of the 53rd Annual ACM SIGACT Symposium on Theory
  of Computing}, pages 859--869, 2021.

\bibitem[vdBLN{\etalchar{+}}20]{van2020bipartite}
Jan van~den Brand, Yin-Tat Lee, Danupon Nanongkai, Richard Peng, Thatchaphol
  Saranurak, Aaron Sidford, Zhao Song, and Di~Wang.
\newblock Bipartite matching in nearly-linear time on moderately dense graphs.
\newblock In {\em 2020 IEEE 61st Annual Symposium on Foundations of Computer
  Science (FOCS)}, pages 919--930. IEEE, 2020.

\bibitem[vdBLSS20]{van2020solving}
Jan van~den Brand, Yin~Tat Lee, Aaron Sidford, and Zhao Song.
\newblock Solving tall dense linear programs in nearly linear time.
\newblock In {\em Proceedings of the 52nd Annual ACM SIGACT Symposium on Theory
  of Computing}, pages 775--788, 2020.

\bibitem[Win13]{windisch2013m}
G{\"u}nther Windisch.
\newblock {\em M-matrices in Numerical Analysis}.
\newblock Springer-Verlag, 2013.

\end{thebibliography}

\end{document}